\documentclass[10pt]{amsart} % [draft] to get overfull box printing
\usepackage[pdftex]{graphicx}
\usepackage[margin=3cm]{geometry}
\usepackage{amsmath, amssymb,amsthm}
\usepackage{mathrsfs} % get script L, H
\usepackage[shortlabels]{enumitem} % [shortlabels] supports changing enumerate labels easily
\usepackage{bbm} % get mathbb1
\usepackage{tikz,pgfplots}
\usepackage{stmaryrd} % get \llbracket,\rrbracket
\usetikzlibrary{bending,math,decorations.pathreplacing,snakes}
\usepackage{xcolor}
\usepackage{hyperref}  % KEEP THIS PACKAGE LAST
\hypersetup{
	colorlinks,
	citecolor=blue,
	filecolor=blue,
	linkcolor=blue,
	urlcolor=blue
}

\newtheorem{thm}{Theorem}[section]
\newtheorem{prop}[thm]{Proposition}
\newtheorem{lem}[thm]{Lemma}
\newtheorem{cor}[thm]{Corollary}
\theoremstyle{definition}

\newtheorem{rmk}{Remark}[section]

\newtheorem*{claim*}{Claim}

\newcommand{\comments}[1]{}
\newcommand{\N}{\mathbb{N}}
\newcommand{\Z}{\mathbb{Z}}
\newcommand{\R}{\mathbb{R}}
\newcommand{\C}{\mathbb{C}}

\newcommand{\K}{K}
\renewcommand{\Im}{\operatorname{\mathrm{Im}}}
\renewcommand{\Re}{\operatorname{\mathrm{Re}}}
\DeclareRobustCommand{\Chi}{{\mathpalette\irchi\relax}}
\newcommand{\irchi}[2]{\raisebox{\depth}{$#1\chi$}} % inner command, used by \rchi
\renewcommand{\epsilon}{\varepsilon}
\renewcommand{\hat}{\widehat}
\renewcommand{\tilde}{\widetilde}
\renewcommand{\bar}{\overline}

\newcommand{\biggp}[1]{\bigg(#1\bigg)}

\newcommand{\intbrr}[1]{\llbracket#1\rrbracket}

\newcommand\numberthis{\addtocounter{equation}{1}\tag{\theequation}}
\numberwithin{equation}{section}

\begin{document}

\title[]{Pointwise Weyl law for graphs from quantized interval maps}
\date{\today}
\author[]{Laura Shou}
\address{\parbox{\linewidth}{Department of Mathematics, Princeton University, Princeton, NJ 08544\\
\textit{Present address: School of Mathematics, University of Minnesota, Minneapolis, MN 55455}\\}}
\email{lshou@math.princeton.edu}

\begin{abstract}
We prove an analogue of the pointwise Weyl law for families of unitary matrices obtained from quantization of one-dimensional interval maps. This quantization for interval maps was introduced by Pako\'nski \textit{et al.}\! in \cite{pzk} as a model for quantum chaos on graphs.  Since we allow shrinking spectral windows in the pointwise Weyl law,  as a consequence we obtain for these models a strengthening of the quantum ergodic theorem from Berkolaiko \textit{et al.}\! \cite{qgraphs}, and show in the semiclassical limit that a family of randomly perturbed quantizations has approximately Gaussian eigenvectors. We also examine further the specific case where the interval map is the doubling map.
\end{abstract}

\maketitle

\section{Introduction}

Quantum graphs have been used as models of idealized one-dimensional structures in physics for many decades, and more recently, as simplified models for studying complex phenomena such as Anderson localization and quantum chaos \cite{book}. The first evidence for quantum chaotic behavior in quantum graphs was given by Kottos and Smilansky \cite{ks1,ks2}, who showed numerically that the spectral statistics of certain families of quantum graphs behave like those of a random matrix ensemble, and further investigated this relationship using an exact trace formula.
In view of the Bohigas--Giannoni--Schmidt conjecture \cite{bgs}, this led to the investigation of quantum graphs as a model for quantum chaos. Additional results regarding convergence of spectral statistics to those of random matrix theory include \cite{bsw2,bsw1,ga1,ga2,tanner}, among others.

In this article, we will consider a quantization method for certain ergodic piecewise-linear 1D interval maps $S:[0,1]\to[0,1]$, introduced by Pako\'nski, \.{Z}yczkowski, and Ku\'{s} in \cite{pzk} as a model for quantum chaos on graphs.
Precise conditions for these interval maps will be described in Section~\ref{sec:main}, but one can consider just the doubling map, $S(x)=2x\pmod{1}$, as a model example.
The quantization method associates to $S$ a family of unitary matrices $U_n$, where $U_n$ is an $n\times n$ matrix, and $n\in\N$ is taken in a subset of allowable dimensions. These unitary matrices can be viewed as describing quantum evolution on directed graphs.
The quantization is done in two steps: one first \emph{discretizes} the map $S$, producing a family of Markov transition matrices $P_n$, where $P_n$ is an $n\times n$ matrix corresponding to a classical Markov chain on a graph with $n$ states.
If $P_n$ is \emph{unistochastic}, meaning there is a unitary matrix $U_n$ with the entrywise relation $|(U_n)_{xy}|^2=(P_n)_{xy}$, then the $U_n$ are called a \emph{quantization} of the classical map $S$.
They are considered ``quantizations'' in the sense they satisfy a classical-quantum correspondence principle (Egorov theorem \cite{qgraphs}), which relates unitary conjugation by the $U_n$ to the classical dynamics of $S$  in the limit as an effective semiclassical parameter, in this case the reciprocal of the dimension, $1/n$, tends to zero.

The matrices $U_n$ will be sparse and can in principle be very simple and non-random, such as the particular quantization for the doubling map given in \eqref{eqn:doubling-pu}.
Surprisingly then, if the classical map $S$ is chaotic, 
then it appears for large $n$ the resulting $U_n$'s
tend to have level spacings that numerically look Wigner--Dyson \cite{pzk,tanner2,tanner}, as well as eigenvector coordinates that numerically look Gaussian (Figure~\ref{fig:num}). This is despite their simple, sparse structure uncharacteristic of a typical CUE Haar unitary matrix.
This behavior is however consistent with major  conjectures in quantum chaos, that quantum systems corresponding to classically chaotic 
ones should exhibit random matrix ensemble spectral statistics (BGS conjecture \cite{bgs}) and have eigenvectors that behave like Gaussian random waves \cite{berry} in the semiclassical limit.

More specifically, as investigated in \cite{pzk,tanner2,tanner}, if the graphs correspond to classically chaotic systems, then  when averaged over a choice of phase in the quantization, the spectral properties of these $U_n$, such as the level spacings, spectral rigidity, and spectral form factor, 
appear to numerically behave like those of CUE random matrices. 
As for eigenvector statistics, quantum ergodicity for these graphs with classically ergodic $S$ was proved by Berkolaiko, Keating, and Smilansky in \cite{qgraphs}. 
They showed that in the large dimension limit, nearly all eigenvectors of $U_n$ equidistribute over their coordinates: for sequences of allowable dimensions $n$,
there is a sequence of sets $\Lambda_n\subseteq\intbrr{1:n}\equiv\intbrr{n}:=\{1,\ldots,n\}$ with $\lim_{n\to\infty}\frac{\#\Lambda_n}{n}=1$ so that for all sequences $(j_n)_{n}$ with $j_n\in \Lambda_n$, and appropriate quantum observables $O_n(\phi)$,
\begin{equation}\label{eqn:qe1}
\lim_{n\to\infty}\langle \psi^{(n,j_n)},O_n(\phi)\psi^{(n,j_n)}\rangle = \int_0^1\phi(x)\,dx,
\end{equation}
where $\psi^{(n,j)}$ is the $j$th eigenvector of $U_n$. 
This is the analogue for these graphs of the Shnirelman--Zelditch--de-Verdi\`{e}re quantum ergodic theorem \cite{shnirelman,zelditch,verdiere}, 
which was originally stated for ergodic flows on compact Riemannian manifolds. Quantum ergodicity has also been extended to other settings such as torus maps \cite{bdb,KurlbergRudnick0, KurlbergRudnick, MarklofOKeefe,ZelditchTori} and other graphs \cite{anantharaman,qe-exp,alm,as};
see also \cite{icm} for an overview and additional references.

In addition to the equidistribution from the quantum ergodic theorem, eigenfunctions from a classically ergodic system are expected to follow Berry's random wave conjecture \cite{berry}, which asserts that the eigenfunctions should behave like Gaussian random waves in the large eigenvalue limit. For graphs, instead of the large eigenvalue limit, one considers as usual the large dimension limit.
In this limit, \cite{phys,phys2} used supersymmetry methods to study the eigenfunction statistics for quantum graphs, specifically Gaussian moments, in view of the random wave conjecture.

In the specific discrete models from interval maps that we consider, 
one expects that the empirical distribution of the coordinates $\{\psi^{(n,j)}_x\}_{x=1}^n$ of a normalized eigenvector of $U_n$ should behave like a random complex Gaussian $N_\C(0,\frac{1}{n})$ for most eigenvectors. This is consistent with both the random matrix ensemble behavior and the random wave conjecture. 
As an example, consider the doubling map $S(x)=2x\pmod{1}$.
For $n\in2\N$, the Markov matrices $P_n$ along with a particularly simple unitary quantization $U_n$, can be taken as,
\begin{align}\label{eqn:doubling-pu}
P_n=\frac{1}{2}
\left(
\begin{smallmatrix}1&1 & & & &&\\
& &1 & 1 &&&\\
&&&&\ddots &&\\
&&&&&1&1\\
1&1&&&&&\\
&&1&1&&&\\
&&&&\ddots&&\\
&&&&&1&1
\end{smallmatrix}\right),\qquad
U_n=\frac{1}{\sqrt{2}}
\left(
\begin{smallmatrix}1&-1 & & & &&\\
& &1 & -1 &&&\\
&&&&\ddots &&\\
&&&&&1&-1\\
1&1&&&&&\\
&&1&1&&&\\
&&&&\ddots&&\\
&&&&&1&1
\end{smallmatrix}\right),
\end{align}
where the non-specified entries are all zeros.
Numerically, for large $n$, the vast majority of eigenvectors of the $U_n$ above have coordinate value distributions that look complex Gaussian $N_\C(0,\frac{1}{n})$. Typical histograms for the coordinates of an eigenvector of $U_n$ from \eqref{eqn:doubling-pu} are shown in Figure~\ref{fig:num}.

\begin{figure}[!ht]
\begin{center}
\includegraphics[height=1.4in]{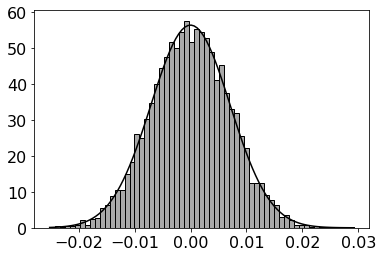}\quad
\includegraphics[height=1.4in]{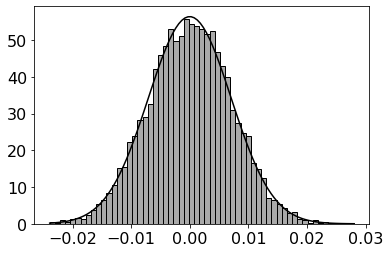}\quad
\includegraphics[height=1.4in]{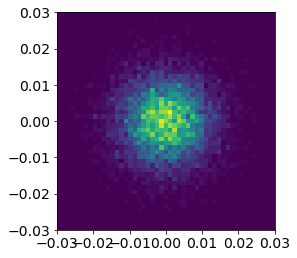}
\end{center}
\caption{Plots for a randomly chosen eigenvector $\psi$ (this one with eigenvalue $-0.3061126+0.9519953i$, chosen uniformly at random out of all eigenvectors) for $n=10\,000$ and $U_n$ in \eqref{eqn:doubling-pu}. 
Left: Histogram of the values $(\Re\psi_x)_{x=1}^{10\,000}$ plotted against the pdf of the real Gaussian ${N}(0,\frac{1}{20\,000})$. Center: Histogram of the values $(\Im\psi_x)_{x=1}^{10\,000}$ plotted against the pdf of ${N}(0,\frac{1}{20\,000})$.
Right: 2D histogram in $\C$ of the coordinates $(\psi_x)_{x=1}^{10\,000}$. Since this is fairly spherically symmetric, the overall choice of phase for the eigenvector does not significantly impact the shapes of the first two plots.
} \label{fig:num}
\end{figure}

Motivated by the above, we will study the eigenvectors of such unitary quantizations $U_n$ constructed from allowable interval maps $S$ by proving a \emph{pointwise Weyl law}, which consists of estimates on the diagonal elements of spectral projection matrices. 
Because we allow for shrinking spectral windows in the pointwise Weyl law, this will have two additional implications concerning eigenvectors. First, we will be able to strengthen the quantum ergodic theorem from \cite{qgraphs}, and second, we will be able to construct random quantizations $V_n$ of $S$ whose eigenvectors have the approximately Gaussian value statistics--while we do not prove Gaussian behavior for the eigenvectors of the original $U_n$ (except in a very special case where the eigenspaces end up highly degenerate, see Section~\ref{sec:doubling-2}), we will show there are many
random quantizations $V_n$ near $U_n$ that do have the desired Gaussian eigenvector coordinates. These quantizations $V_n$ are not quantizations in the strict sense of $|(V_n)_{xy}|^2=(P_n)_{xy}$ from \cite{pzk,qgraphs}, but they will satisfy $|(V_n)_{xy}|^2=(P_n)_{xy}+o(1)$ as well as an Egorov theorem, so they are still quantizations of $S$ in the sense that they satisfy a classical-quantum correspondence principle, recovering the classical dynamics in the semiclassical limit $n\to\infty$.

Traditionally, a pointwise Weyl law gives the leading order asymptotics of the spectral projection kernel $\mathbbm{1}_{(-\infty,t]}(-\Delta+V)(x,x)$, for $x$ in $M$ a compact Riemannian manifold.
For the unitary matrices $U_n$, we look at the spectral projection onto arcs on the unit circle, $P^I=\sum_{j:\theta^{(n,j)}\in I}|\psi^{(n,j)}\rangle\langle\psi^{(n,j)}|$ where $I\subseteq\R/(2\pi\Z)$. Then making use of the  little-o asymptotic notation, a pointwise Weyl law analogue would 
be a statement of the form $
\sum_{j:\theta^{(n,j)}\in I(n)}|\psi_x^{(n,j)}|^2=\frac{|I(n)|}{2\pi}(1+o(1))
$ for $n\to\infty$ and appropriate intervals $I(n)$. We will show this holds for sequences of intervals $I(n)$ shrinking at certain rates, and for at least $n(1-o(|I(n)|))$ coordinates $x$.
The coordinates for which this statement may not hold correspond to short periodic orbits in the graphs corresponding to $S$.

We then present the two applications of this pointwise Weyl law. 
The first is the strengthening of the quantum ergodic theorem to apply to sets of eigenvectors in the bins $\{\psi^{(n,j)}:\theta^{(n,j)}\in I(n)\}$ with shrinking $I(n)$. This will mean that equation \eqref{eqn:qe1} must hold for more $j_n$, specifically limiting density one sets within the shrinking density zero sets $I(n)$.
The second concerns random perturbations of the matrix $U_n$ to produce a family of unitary random matrices $V_n(\beta^{[n]})$, with $\beta^{[n]}$ the random parameter, whose eigenvectors have the approximately Gaussian $N_\C(0,\frac{1}{n})$ eigenvector statistics. 
These eigenvectors will also tend to satisfy a version of quantum unique ergodicity (QUE), a notion introduced by Rudnick and Sarnak in \cite{RudnickSarnak}, and where all eigenvectors are considered in the limit \eqref{eqn:qe1}, rather than just those in a sequence of limiting density one sets.

\subsubsection*{Acknowledgements} 
The author would like to thank Ramon van Handel for suggesting this problem and providing many helpful discussions and feedback, and for pointing out the simpler proof of Proposition~\ref{prop:u2k}. 
The author would also like to thank Peter Sarnak for helpful discussion and suggesting the use of the Beurling--Selberg function as the particular smooth approximation.
This material is based upon work supported by the National Science Foundation Graduate Research Fellowship under Grant No. DGE-2039656. Any opinion, findings, and conclusions or recommendations expressed in this material are those of the author and do not necessarily reflect the views of the National Science Foundation.
%https://www.nsf.gov/pubs/2012/nsf12062/nsf12062.jsp

\section{Set-up and main results}\label{sec:main}

\subsection{Set-up} \label{subsec:setup}
Here we state the assumptions on the map $S$ and matrices $P_n$.
%\begin{assumption}\label{a:S}
Let $S:[0,1]\to[0,1]$ be a piecewise-linear map that satisfies the following conditions:
\begin{enumerate}[(i)]
\item $S$ is (Lebesgue) measure-preserving, $\mu(A)=\mu(S^{-1}(A))$ for any measurable set $A$.
\item There exists a partition of $[0,1]$ into $M_0$ equal intervals (called atoms) $A_1,\ldots,A_{M_0}$, with $S$ linear on each atom $A_j$. This partition will be denoted by $\mathcal{M}_0$, and the collection of endpoints of the atoms $(A_j)_j$ will be denoted by $\mathcal{E}_0=\{0,\frac{1}{M_0},\frac{2}{M_0},\ldots,\frac{M_0-1}{M_0},1\}$.
\item  The linear segments of $S$ begin and end in the grid $\mathcal{E}_0\times \mathcal{E}_0$; more formally, the left and right limits of $S$ satisfy $\lim_{x\to e_0^\pm}S(x)\in\mathcal{E}_0$ for $e_0\in\mathcal{E}_0$.  With (i) and (ii), this ensures the slope of $S$ on each atom must be an integer. For convenience, also assume $S(e_0)$ takes one of the values of these one-sided limits.
\item The absolute value of the slope of $S$ on each atom is at least two, i.e. the slope is never $\pm1$.
\end{enumerate}
%\end{assumption}
Conditions (i), (ii), and (iii) are essentially the same as in \cite{pzk,qgraphs}. Condition (iv) is there instead of the ergodicity assumption. It allows for some non-ergodic $S$ such as those corresponding to block matrices of various ergodic maps. 
Two examples of allowable ergodic  $S$ are the doubling map and the ``four legs map'' shown in Figure~\ref{fig:maps}.
In general, conditions for ergodicity of $S$ would follow from results on piecewise expanding Markov maps, see for example Chapter III in the textbook \cite{mane}.

\begin{figure}[!ht]
\centering
\begin{tikzpicture}[scale=.7]
\def\shift{8} % the amount to shift the right graph by
\foreach \g in {0,\shift}{
\draw (\g,0)--(\g+4,0)--(\g+4,4)--(\g,4)--cycle;
\draw[dotted] (\g+2,0)--(\g+2,4);
\draw[dotted] (\g,2)--(\g+4,2);
\node [below left] at (\g,0) {$0$};
\node [below] at (\g+4,0) {$1$};
\node [left] at (\g,4) {$1$};
}
\draw[dotted] (\shift+1,0)--(\shift+1,4);
\draw[dotted] (\shift+3,0)--(\shift+3,4);
\draw[dotted] (\shift,1)--(\shift+4,1);
\draw[dotted] (\shift,3)--(\shift+4,3);
\draw (\shift,0)--(\shift+1,2);
\draw (\shift+1,0)--(\shift+2,4);
\draw (\shift+2,0)--(\shift+3,4);
\draw (\shift+3,2)--(\shift+4,4);
% doubling map
\draw (0,0)--(2,4);
\draw (2,0)--(4,4);
\end{tikzpicture}
\caption{The doubling map (left) and ``four legs map'' (right). For the doubling map $M_0=L_0=2$, while for the four legs map $M_0=L_0=4$.}\label{fig:maps}
\end{figure}
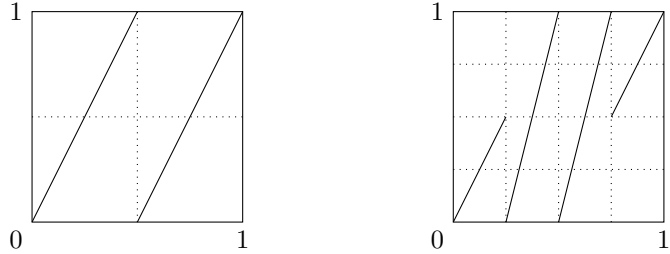

For $n\in M_0\Z$, partition $[0,1]$ into $n$ equal atoms, $E_x=(\frac{x-1}{n},\frac{x}{n})$ for $x=1,\ldots,n$, and define the corresponding $n\times n$ Markov transition matrix $P_n$ by
\begin{equation}
(P_n)_{xy}=\begin{cases}
0,&S(E_x)\cap E_y=\emptyset\\
\frac{1}{|S'(z)|},& S(E_x)\cap E_y\ne\emptyset,\text{ any }z\in E_x
\end{cases}.
\end{equation}
The matrix $P_n$ looks at where $S$ sends an atom $E_x$, and assigns a uniform probability $\frac{1}{|S'(z)|}$ to each atom $E_y$ that $S$ can reach from $E_x$.
To generate the family of corresponding unitary matrices $U_n$ as done in \cite{pzk,qgraphs}, it is required that the $P_n$ be \emph{unistochastic}, so that there are unitary matrices $U_n$ with the entrywise relation $|(U_n)_{xy}|^2=(P_n)_{xy}$. In general, characterizing which bistochastic matrices are unistochastic is difficult; however see \cite{pzk, unistochastic, qgraphs} for some conditions and examples. 
Note that the relation $|(U_n)_{xy}|^2=(P_n)_{xy}$ does not uniquely define $U_n$ if it exists, as one can always add additional phases without changing unitarity or the entrywise relation. For example, given any $\Phi\in[0,2\pi)^n$ and defining the diagonal matrix $e^{i\Phi}:=\operatorname{diag}(e^{i\Phi_1},\ldots,e^{i\Phi_n})$, then $e^{i\Phi}U_n$ is also unitary and satisfies the same entrywise norm-squared relation.

Finally, let $L_0$ be the least common multiple of the slopes in $S$, and let ${\K}(n)$ be the largest power of $L_0$ that divides $n/M_0$, so $n=M_0L_0^{{\K}(n)}r$ and $r$ does not contain any factors of $L_0$. 
The technical purpose of ${\K}(n)$ will be to keep track of how many powers $\ell$ of $S$ we can take, while still ensuring $S^\ell$ behaves nicely with the partition into $n$ atoms. This $K(n)$ can also be thought of as an Ehrenfest time, cf. Remark~\ref{rmk:mainthm}(iii).

\subsection{Pointwise Weyl law and overview}

With the above definitions, we state the main result:
\begin{thm}[pointwise Weyl law analogue]\label{thm:uweyl}
Let $S:[0,1]\to[0,1]$ satisfy assumptions (i)--(iv). 
Consider a sequence $(n_k)_k$ so that ${\K}(n_k)\to\infty$, and suppose each $n_k\times n_k$ Markov matrix $P_{n_k}$ is unistochastic with corresponding unitary matrix $U_{n_k}$.
Let $(I(n_k))$ be a sequence of intervals in $\R/(2\pi \Z)$ satisfying
\begin{equation}\label{eqn:km}
{|I(n_k)| {\K}(n_k)}\to\infty,\quad\text{as }k\to\infty.
\end{equation} 
Then denoting the eigenvalues and eigenvectors of $U_{n_k}$ by $(e^{i \theta^{(n_k,j)}})_j$ and $(\psi^{(n_k,j)})_j$ respectively, there is a sequence of subsets $G_{n_k}\subseteq \intbrr{n_k}$ with sizes $\#G_{n_k}=n_k(1-o(|I(n_k)|))$ so that for all $x\in G_{n_k}$,
\begin{equation}\label{eqn:I}
\sum_{j:\theta^{(n_k,j)}\in I(n_k)}|\psi_x^{(n_k,j)}|^2=\frac{|I(n_k)|}{2\pi}(1+o(1)),\quad\text{as }k\to\infty,
\end{equation}
where the error term $o(1)$ can be taken to depend only on $n_k$, $|I(n_k)|$, and ${\K}(n_k)$, 
and is independent of $x\in G_{n_k}$. Additionally, $G_{n_k}$ can actually be chosen independent of $I(n_k)$ and $|I(n_k)|$.
% indep of |I(n)|: can just use worst case from |I(n)|K\to\infty
\end{thm}
\begin{rmk}
\makeatletter
\hyper@anchor{\@currentHref}%
\makeatother
\label{rmk:mainthm}
\begin{enumerate}[(i)]
\item 
Equation~\eqref{eqn:I} cannot in general be improved to hold for all coordinates $x\in\intbrr{n_k}$, as Section~\ref{sec:coordfail} will show.
The coordinates $x$ that we exclude from $G_{n_k}$ correspond to those with short periodic orbits in the graphs associated to $P_{n_k}$ and $U_{n_k}$. This is reminiscent of the relationship between geodesic loops and the size of the remainder in the Weyl law \cite{DG,Ivrii} and pointwise Weyl law \cite{Safarov,SoggeZelditch,CanzaniGalkowski}, in the usual setting on manifolds.

\item The condition \eqref{eqn:km} that ${\K}(n_k)|I(n_k)|\to\infty$ is essentially optimal under the stated hypotheses, as we will see by considering the doubling map when $n=2^K$ (Section~\ref{sec:doubling-2}). In this case, ${\K}(n)=\log_2n$, and if $|I(n)|<\frac{\pi}{2\log_2n}$, then one can take an interval that avoids the spectrum entirely and thus produces a spectral projection matrix filled with zeros.

\item This condition  that $|I(n_k)|$ does not shrink too fast appears from error terms from only considering powers of $U_{n_k}$ up to an Ehrenfest time ${\K}(n_k)\sim\log n_k$. 
This time is a common obstruction in semiclassical problems, and even in these discrete models, our analysis does not go beyond this time.
If the lengths $|I(n_k)|$ are larger than the bare minimum required to satisfy \eqref{eqn:km}, then more precise remainder terms than just $o(1)$ are obtained from the proof, cf. equation \eqref{eqn:close}.
\end{enumerate}
\end{rmk}
The proof details of Theorem~\ref{thm:uweyl} will be specific to our discrete case, where we have sparse matrices $U_{n_k}$ and can analyze matrix powers and paths in finite graphs.
We will start by just taking a smooth approximation of the indicator function of the interval $I(n_k)$, and estimating the left side of \eqref{eqn:I} by a Fourier series in terms of powers of $U_{n_k}$. However, the properties of $S$ ensure that we understand powers of  $U_{n_k}$ well up to time ${\K}(n_k)$. This allows us to identify and exclude the few coordinates $x$ that have short loops before a set cut-off time. Using properties of the powers of $U_{n_k}$ again, the remaining coordinates will then produce small enough Fourier coefficients that \eqref{eqn:I} holds.

Summing \eqref{eqn:I} over all $x$ (separating $x\in G_{n_k}$ from $x\not\in G_{n_k}$) produces a Weyl law analogue that counts the number of eigenvalues in a bin.
\begin{cor}[Weyl law analogue]\label{cor:weyl}
Let $S$, $(n_k)_k$, $U_{n_k}$, and $I(n_k)$ be as in Theorem~\ref{thm:uweyl}, including \eqref{eqn:km}. Then as $k\to\infty$,
\begin{equation}
\#\{j:\theta^{(n_k,j)}\in I(n_k)\} = n_k\frac{|I(n_k)|}{2\pi}(1+o(1)),
\end{equation}
where the remainder term depends on $|I(n_k)|$ but is independent of the particular location of $I(n_k)$.
\end{cor}

In the following subsections, we discuss implications of Theorem~\ref{thm:uweyl} on eigenvectors of $U_n$. We present the two applications, the first a strengthening of the quantum ergodic theorem for this model, and the second a construction of random perturbations of $U_n$ with approximately Gaussian eigenvectors. For the first application, using Theorem~\ref{thm:uweyl} with shrinking intervals $|I(n_k)|$, rather than the usual local Weyl law, in the standard proof of quantum ergodicity naturally produces a stronger quantum ergodicity statement. For the second, we take random unitary rotations of bins of eigenvectors, and apply results on the distribution of random projections from \cite{DiaconisFreedman,ChatterjeeMeckes,Meckes} to show the resulting eigenvectors have approximately Gaussian value statistics.

\subsection{Application to quantum ergodicity in bins}
To state a quantum ergodic theorem, we first define quantum observables as in \cite{qgraphs}, as discretized versions of a classical observable $h\in L^2([0,1])$. Given $n\in\N$ and $h\in L^2([0,1])$, define its quantization $O_n(h)$ to be the $n\times n$ diagonal matrix with entries
\begin{equation}
O_{n}(h)_{xx}=\frac{1}{|E_x|}\int_{E_x}h(z)\,dz=n\int_{E_x}h(z)\,dz.
\end{equation}
Note that $\frac{1}{n}\operatorname{tr}O_n(h)=\int_0^1h$, the analogue of the local Weyl law.
Quantum ergodicity for this model, as proved in \cite{qgraphs}, states there is a sequence of sets $\Lambda_{n_k}\subseteq\intbrr{n_k}$ with $\lim_{n_k\to\infty}\frac{\#\Lambda_{n_k}}{n_k}=1$ such that for all sequences $(j_{n_k})_{k}$, $j_{n_k}\in \Lambda_{n_k}$ and $h\in C([0,1])$,
\begin{equation}\label{eqn:qe}
\lim_{k\to\infty}\langle \psi^{(n_k,j_{n_k})},O_{n_k}(h)\psi^{(n_k,j_{n_k})}\rangle=\int_0^1h(x)\,dx.
\end{equation}
This is equivalent to the decay of the quantum variance, 
\[
V_{n_k}:=\frac{1}{n_k}\sum_{j=1}^{n_k}\left|\langle \psi^{(n_k,j)},O_{n_k}(h)\psi^{(n_k,j)}\rangle-\int_0^1h(x)\,dx\right|^2\to0,
\]
as $k\to\infty$. Using Theorem~\ref{thm:uweyl} and an Egorov property from \cite{qgraphs}, we will prove the following concerning quantum ergodicity in bins.
\begin{thm}[Quantum ergodicity in bins]\label{thm:qe}
Let $S$ satisfy (i)--(iv) and also be ergodic. Let $(n_k)_k$, $U_{n_k}$, and $I(n_k)$ be as in Theorem~\ref{thm:uweyl}, including \eqref{eqn:km}. Then for any Lipschitz $h:[0,1]\to\C$,
\begin{equation}\label{eqn:qvariance}
\frac{1}{\#\{j:\theta^{(n_k,j)}\in I(n_k)\}}\sum_{j:\theta^{(n_k,j)}\in I(n_k)}\left|\langle \psi^{(n_k,j)},O_{n_k}(h)\psi^{(n_k,j)}\rangle-\int_0^1h(x)\,dx\right|^2\to0,
\end{equation}
as $k\to\infty$.
\end{thm}
This decay of the quantum variance in a bin implies there is a sequence of sets $\Lambda_{n_k}\subseteq\{j:\theta^{(n_k,j)}\in I(n_k)\}$ with $\frac{\#\Lambda_{n_k}}{\#\{j:\theta^{(n_k,j)}\in I(n_k)\}}\to1$ such that \eqref{eqn:qe} holds for all sequences $(j_{n_k})_k$, $j_{n_k}\in\Lambda_{n_k}$ and continuous $h:[0,1]\to\C$.
Since we allow $|I(n)|\to0$, the bin sizes are $o(n)$ (density 0) by the Weyl law analogue, and so Theorem~\ref{thm:qe} guarantees that the  density 0 subsequence excluded from the original quantum ergodic theorem cannot accumulate too strongly in a region $I(n)$ of the unit circle.

\subsection{Application to random quantizations with Gaussian eigenvectors}\label{subsec:rquant}
The second application of Theorem~\ref{thm:uweyl} will be to construct random perturbations of $U_n$ with eigenvectors whose values look approximately Gaussian. 
To construct the random perturbations, we use the estimates on the spectral projection matrix combined with results on low-dimensional projections from \cite{DiaconisFreedman,Meckes,ChatterjeeMeckes} to first prove:

\begin{thm}[Gaussian approximate eigenvectors]\label{thm:r-eigenvectors}
Let $S$, $(n_k)_k$, $U_{n_k}$, and $I(n_k)$ be as in Theorem~\ref{thm:uweyl}, in particular assume \eqref{eqn:km} holds. Denote the eigenvalues and eigenvectors of $U_{n_k}$ by $(e^{i\theta^{(n_k,j)}})_j$ and $(\psi^{(n_k,j)})_j$ respectively.
Then letting $\phi^{(n_k)}$ be a unit vector chosen randomly according to Lebesgue measure from $\operatorname{span}(\psi^{(n_k,j)}:\theta^{(n_k,j)}\in I(n_k))$, the empirical distribution $\mu^{(n_k)}=\frac{1}{n_k}\sum_{j=1}^{n_k}\delta_{\sqrt{n_k}\phi_j^{(n_k)}}$ of the scaled coordinates
\[
\sqrt{n_k}\phi^{(n_k)}_1,\sqrt{n_k}\phi^{(n_k)}_2,\ldots,\sqrt{n_k}\phi^{(n_k)}_{n_k},
\]
converges weakly in probability to the standard complex Gaussian $N_\C(0,1)$ as $k\to\infty$. In fact, there are absolute numerical constants $C,c>0$ 
so that for any $f:\C\to\C$ bounded with Lipschitz constant $\|f\|_\mathrm{Lip}:=\sup_{x\ne y}\frac{|f(x)-f(y)|}{|x-y|}\le L$ and $\varepsilon>0$, there is $k_0>0$ so that for $k\ge k_0$,
\begin{equation}
\mathbb{P}\left[\left|\int f(x)\,d\mu^{(n_k)}(x)-\mathbb{E}f(Z)\right|>\varepsilon\right]\le C\exp\left(-c\varepsilon^2n_k|I(n_k)|/L^2\right),
\end{equation}
where $Z\sim N_\C(0,1)$.
\end{thm}

We then construct random perturbations $V_{n_k}(\beta^{[n_k]})$ of $U_{n_k}$ by grouping the eigenvalues of $U_{n_k}$ into bins and randomly rotating the eigenvectors within each bin.
The random parameter $\beta^{[n_k]}$ represents the choice of random rotations.
This idea of rotating small sets of eigenvectors was used in different models in \cite{zelditch-random,VanderKam,maples,randomQUE} to construct random orthonormal bases with quantum ergodic or quantum unique ergodic properties. 
In our setting, 
Theorem~\ref{thm:r-eigenvectors} will show that the coordinates of these randomly rotated eigenvectors look approximately Gaussian. The matrices $V_{n_k}(\beta^{[n_k]})$ also satisfy the entrywise relations $|(V_{n_k}(\beta^{[n_k]}))_{xy}|^2=(P_{n_k})_{xy}+o(1)$ as well as a weaker Egorov property relating them to the classical dynamics, so that they can be viewed as a quantization of the classical map $S$. Thus while we do not prove
approximate Gaussian behavior for the quantizations $U_{n_k}$ with $|(U_{n_k})_{xy}|^2=(P_{n_k})_{xy}$
(except in a very special case where the eigenspaces end up highly degenerate, see Section~\ref{sec:doubling-2}),
we prove the approximate Gaussian eigenvector behavior for the family of random matrices $V_{n_k}(\beta^{[n_k]})$, which are alternative  quantizations of the original classical dynamics of $S$.

Note also that $S$ is not required to be ergodic here. In particular, we can take the direct sum of two ergodic maps $S_1$ and $S_2$, whose resulting block matrix $U_{n_k}$ has eigenvectors localized on just half of the coordinates. Then $U_{n_k}$ will not have equidistributed or Gaussian eigenvectors, though the randomly perturbed matrices $V_{n_k}(\beta^{[n_k]})$ still will.

\vspace{1mm}
In what follows, we will continue using the notation $V_{n_k}(\beta^{[n_k]})$ for this family of randomly perturbed matrices, with $\beta^{[n_k]}$ as the random parameter which will represent the choices for the random rotations (see Section~\ref{subsec:proof-gvectors} for further details).

\begin{thm}[Random quantizations with Gaussian eigenvectors]\label{cor:gvectors}
Let $S$ satisfy (i)--(iv), and let $(n_k)_k$ be a sequence with ${\K}(n_k)\to\infty$ and with each Markov matrix $P_{n_k}$  unistochastic. Then there exists a family of random unitary matrices $V_{n_k}(\beta^{[n_k]})$ in some probability spaces $(\Omega_{n_k},\mathbb{P}_{n_k})$, with the following properties:
\begin{enumerate}[(a)]
\item  $V_{n_k}(\beta^{[n_k]})$ is a small perturbation of $U_{n_k}$, as in $\sup_{\beta^{[n_k]}}\|V_{n_k}(\beta^{[n_k]})-U_{n_k}\|=o(1)$. Additionally, for every $
\beta^{[n_k]}$, $V_{n_k}$ satisfies an Egorov property; for Lipschitz $h:[0,1]\to\C$,
\[
\|V_{n_k}O_{n_k}(h)V_{n_k}^{-1}-O_{n_k}(h\circ S)\|=o(1)\cdot{\|h\|_\mathrm{Lip}}.
\]

\item (Gaussian coordinates).
There is a sequence of sets $\Pi_{n_k}\subseteq\Omega_{n_k}$ with $\mathbb{P}[\Pi_{n_k}]\to1$ with the following property: Let $(\tilde{V}_{n_k})_k$ be a sequence of matrices with $\tilde{V}_{n_k}\in\Pi_{n_k}$, and let $\tilde{\phi}^{[n_k,j]}$ be the $j$th eigenvector of $\tilde{V}_{n_k}$, and  $\mu^{[n_k,j]}=\frac{1}{n_k}\sum_{x=1}^{n_k}\delta_{\sqrt{n_k}\tilde{\phi}^{[n_k,j]}_x}$ the empirical distribution of the scaled coordinates of $\tilde{\phi}^{[n_k,j]}$. Then for every sequence $(j_{n_k})_k$ with $j_{n_k}\in\intbrr{n_k}$, the sequence $(\mu^{[n_k,j_{n_k}]})_k$ converges weakly to $N_\C(0,1)$ as $k\to\infty$.

\item (QUE).
There is a sequence of sets $\Gamma_{n_k}\subseteq\Omega_{n_k}$ with $\mathbb{P}[\Gamma_{n_k}]\to1$ such that for any sequence of  matrices $(\tilde{V}_{n_k})_k$ with $\tilde{V}_{n_k}\in\Gamma_{n_k}$, the eigenvectors $\tilde{\phi}^{[n_k,j]}$ of $\tilde{V}_{n_k}$ equidistribute over their coordinates. That is, for any sequence $(j_{n_k})_k$ with $j_{n_k}\in\intbrr{n_k}$, and any $h\in C([0,1])$,
\begin{equation}\label{eqn:que2}
\lim_{k\to\infty}\langle\tilde{\phi}^{[n_k,j_{n_k}]},O_{n_k}(h)\tilde{\phi}^{[n_k,j_{n_k}]}\rangle=\int_0^1h(x)\,dx.
\end{equation}

\item For every $\beta^{[n_k]}$, the spectrum of $V_{n_k}(\beta^{[n_k]})$ is non-degenerate.
\item The matrix elements of $V_{n_k}(\beta^{[n_k]})$ satisfy $\sup_{\beta^{[n_k]}}\max_{x,y}\left||V_{n_k}(\beta^{[n_k]})_{xy}|^2-(P_{n_k})_{xy}\right|\to0$ as $k\to\infty$.
\end{enumerate}
\end{thm}

The weak convergence in (b) of the empirical distribution means that the value distribution of the coordinates of the eigenvector $\tilde{\phi}^{[n_k,j]}$ looks complex Gaussian. This is not a statement about the spatial behavior or plot of this vector, but rather describes the idea that if one plots the histogram of the coordinates of $\tilde{\phi}^{[n_k,j]}$ (similarly as done for the eigenvectors in Figure~\ref{fig:num}), then it will tend to look like the density of a complex Gaussian.

\subsection{The doubling map}\label{subsec:doubling}
Finally, in Sections~\ref{sec:doubling-even} and \ref{sec:doubling-2} we study the case when $S$ is the doubling map $2x\pmod{1}$ and the specific quantization $U_n$ is the orthogonal one in \eqref{eqn:doubling-pu}.
We study this case using similar arguments as in the general case, but with stronger estimates from analyzing binary trees and bit shifts specific to the doubling map.
Theorem~\ref{thm:uweyl} will hold with any sequence of even $n\in2\N$, not just those with ${\K}(n)\to\infty$.
Additionally, when $n=2^K$, the spectrum of this specific quantization $U_n$ is degenerate with multiplicities asymptotically $\frac{2^K}{4K}$, and similar arguments as used for Section~\ref{subsec:rquant} will show that most every eigenbasis looks Gaussian (Theorem~\ref{thm:doub2}).

\subsection{Outline}

Section~\ref{sec:plem} contains some lemmas concerning properties of the map $S$ and the corresponding Markov matrices $P_n$. Section~\ref{sec:weylproof} is the proof of the pointwise Weyl law analogue, Theorem~\ref{thm:uweyl}.
The first application, on quantum ergodicity in bins, is proved in Section~\ref{sec:qe}. 
Section~\ref{sec:gaussian} covers the second application on random perturbations of $U_n$ with approximately Gaussian eigenvectors. 
Sections~\ref{sec:doubling-even} and \ref{sec:doubling-2} deal with the specific map the doubling map, especially with the degenerate case of dimension a power of two.
An example where the pointwise Weyl relation \eqref{eqn:I} fails for a particular coordinate choice $x$ is given in Section~\ref{sec:coordfail}.

\section{Properties of the map $S$ and matrices \texorpdfstring{$P_n$}{Pn}}\label{sec:plem}

In this section we gather some results about the relationship between the map $S$ and the Markov matrices $P_n$. The following lemma contains properties from \cite{qgraphs} and \cite{pzk}, stated here for a specific condition involving ${\K}(n)$.

\begin{lem}[powers of $S$, \cite{qgraphs,pzk}]\label{lem:sp}
Let $S:[0,1]\to[0,1]$ be a piecewise-linear map satisfying the assumptions (i)--(iii) from the beginning of Section~\ref{subsec:setup}. Let the partition size be $n\in M_0\Z$ with atoms $E_1,\ldots,E_n$, and let $\mathcal{E}_n$ be the set of all endpoints of the atoms $E_x$. Then for $1\le\ell\le{\K}(n)+1$,
\begin{enumerate}[(a)]
\item $S^\ell$ is linear with integer slope on each atom $E_x$. For endpoints $e\in\mathcal{E}_n$, the right and left limits satisfy $\lim_{y\to e^\pm}S^\ell(y)\in\mathcal{E}_n$, i.e., the start and end of each linear segment live in the grid $\mathcal{E}_n\times\mathcal{E}_n$.
\item If $S^\ell(E_x)\cap E_y\ne\emptyset$, then  $S^\ell(E_x)\supset E_y$. In fact $S^\ell(E_x)$ is a union of several adjacent atoms and some endpoints.
\item $(P_n)_{x\tau_1}(P_n)_{\tau_1\tau_2}\cdots (P_n)_{\tau_{\ell-1}y}\ne0$ iff there exists $z\in E_x$ with $S^\ell(z)\in E_y$ and $S^j(z)\in E_{\tau_j}$ for $j=1,\ldots,\ell-1$. 
\item If $S^\ell(E_x)\cap E_y=\emptyset$ then $(P_n^\ell)_{xy}=0$. If $S^\ell(E_x)\cap E_y\ne\emptyset$, then there is a unique sequence $\tau=(\tau_1,\tau_2,\ldots,\tau_{\ell-1})$ such that $(P_n)_{x\tau_1}(P_n)_{\tau_1\tau_2}\cdots (P_n)_{\tau_{\ell-1}y}\ne0$.
\end{enumerate}
\end{lem}
The condition here with ${\K}(n)$ can be more restrictive than needed in \cite{qgraphs}, but is a concrete example of allowable powers $\ell$ and dimensions $n$. For completeness with these concrete conditions, we include most of the proofs below.
\begin{proof}
\begin{enumerate}[(a)]
\item Both parts are done recursively in $\ell$. For example, suppose $S^{\ell-1}$ is linear on each atom $E_x$. Then for the first part of (a), it suffices to show for each $x$, $S^{\ell-1}(E_x)\subseteq A_j$ for one of the atoms $A_j$ of the ``base'' partition $\mathcal{M}_0$ (the specific $A_j$ can depend on $x$), since then composition with $S$ shows $S^\ell=S\circ S^{\ell-1}$ is linear on $E_x$ since $S$ is linear on $A_j$. 

The inclusion $S^{\ell-1}(E_x)\subseteq A_j$ holds for any $\ell-1\le {\K}(n)$, essentially because this image must avoid all endpoints $e_0\in\mathcal{E}_0$:
We will show the preimage $S^{-(\ell-1)}(\mathcal{E}_0)$ is contained in the endpoints $\mathcal{E}_n$.
Then if we take the open interval (interior) $E_x^\mathrm{o}$, which is in-between endpoints in $\mathcal{E}_n$, then $S^{\ell-1}(E_x^{\mathrm{o}})$ must not contain any points $e_0\in\mathcal{E}_0$, and so must be contained inside just  a single atom $A_j$ as $S^{\ell-1}$ is linear on $E_x$.

For $\ell\ge2$, suppose $S^{-(\ell-2)}(\mathcal{E}_0)\subseteq\frac{1}{M_0L_0^{\ell-2}}\Z$ and let $y$ be in the preimage $S^{-(\ell-1)}(\mathcal{E}_0)$; we will show $y\in\frac{1}{M_0L_0^{\ell-1}}\Z$. Recall  that for $L_0$ the least common multiple of the slopes of $S$ we wrote  $n=M_0L_0^{{\K}(n)}r$, so that for $\ell-1\le {\K}(n)$, then $\frac{1}{n}\Z\supseteq\frac{1}{M_0L_0^{\ell-1}}\Z$, so the above is sufficient to show $y\in\mathcal{E}_n$. Since $S(y)=my+b\in S^{-(\ell-2)}(\mathcal{E}_0)\subseteq\frac{1}{M_0L_0^{\ell-2}}\Z$ for some $m$ and $b$, and we know in fact $m|L_0$ and $b\in\frac{1}{M_0}\Z$ from the definitions on $S$, then $y\in\frac{1}{M_0L_0^{\ell-2}m}\subseteq\frac{1}{M_0L_0^{\ell-1}}\Z$.

\item follows from (a) since the linear segments in $S^\ell$ start at points in $\mathcal{E}_n$ and have integer slopes. 
\item The ($\Leftarrow$) direction is immediate from the definition of $P_n$.  The ($\Rightarrow$) direction follows from the relations
\begin{align*}
S(E_x)\supset E_{\tau_1}, \quad 
S(E_{\tau_1})\supset E_{\tau_2},\quad
\ldots, \quad
S(E_{\tau_{\ell-1}})\supset E_y
\end{align*}
and working backwards, taking $z_{\ell-1}\in E_{\tau_{\ell-1}}$ with $S(z_{\ell-1})\in E_y$, and then $z_j\in E_{\tau_j}$ with $S(z_j)=z_{j+1}$.
\item \begin{sloppypar}
The first part follows from the above inclusions as well; note that if $(P_n)_{x\tau_1}(P_n)_{\tau_1\tau_2}\cdots (P_n)_{\tau_{\ell-1}y}\ne0$ for some sequence $\tau_1,\ldots,\tau_{\ell-1}$, then
$
S^\ell(E_x)\supset S^{\ell-1}(E_{\tau_1})\supset\cdots\supset E_y,
$
so that $S^\ell(E_x)\cap E_y\ne\emptyset$.

The unique path part  is Lemma 2 from \cite{qgraphs}: the proof is to suppose for contradiction there are $z_1,z_2\in E_x$ with $S^\ell(z_1),S^\ell(z_2)\in E_y$ but  with $S^r(z_1)\in E_{\alpha}$ and $S^r(z_2)\in E_{\beta}$ for some $1<r<\ell$ and $\alpha\ne\beta$.
Then pick any $w\in E_y$, and by part (b), there is a  preimage $w_1\in E_{\alpha}$ with $S^{\ell-r}(w_1)=w$ and $w_2\in E_{\beta}$ with $S^{\ell-r}(w_2)=w$. Again by (b) then there are preimages $v_1\ne v_2$ in $E_x$ with $S^r(v_1)=w_1$ and $S^r(v_2)=w_2$. But then $S^\ell(v_1)=S^\ell(v_2)=w$ which contradicts $S^\ell$ being linear by part (a) (with nonzero slope since $S$ is measure-preserving) and thus injective on $E_x$.
\end{sloppypar}
\end{enumerate}
\end{proof}

The next lemma shows how ${\K}(n)$ is used to ensure that small powers of $P_n^\ell$ interact nicely with the partition of $[0,1]$ into $n$ atoms.
\begin{lem}[powers of $P_n$]\label{lem:ppowers}
Assume (i)--(iii) and let $1\le \ell\le {\K}(n)+1$. Then
\begin{equation}\label{eqn:ppowers}
(P_n^\ell)_{xy}=\begin{cases}
0,&S^\ell(E_x)\cap E_y=\emptyset
\\\frac{1}{|(S^\ell)'(z)|},&S^\ell(E_x)\cap E_y\ne\emptyset,\text{ any }z\in E_x
\end{cases}.
\end{equation}
That is, for $1\le\ell\le{\K}(n)+1$, we can compute $P_n^\ell$ by drawing $S^\ell$ and applying the same procedure we used to define $P_n$ from $S$.
\end{lem}
\begin{proof}
From Lemma~\ref{lem:sp}(a), partitioning $[0,1]$ into $M_0\cdot L_0^{\ell-1}$ equal atoms ensures $S^\ell$ is linear on each atom. Since $1\le\ell\le{\K}(n)+1$, then $M_0\cdot L_0^{\ell-1}$ divides $n$ so for these $\ell$, the map $S^\ell$ is linear on each atom of the size $n$ partition and the value $\frac{1}{|(S^\ell)'(z)|}$ is the same for any $z\in E_x$.
The matrix elements of $P_n^\ell$ are
$
(P_n^\ell)_{xy}=\sum_{\tau:x\to y}(P_n)_{x\tau_1}(P_n)_{\tau_1\tau_2}\cdots (P_n)_{\tau_{\ell-1}y}.
$
By Lemma~\ref{lem:sp}(d), for $1\le\ell\le {\K}(n)+1$ and fixed $x,y$, this sum over $\tau$ collapses to either zero or just a single term $(P_n)_{x\tau_1}\cdots (P_n)_{\tau_{\ell-1}y}$. 
If this is nonzero, then by the definition of $P_n$,
\begin{align*}
(P_n^\ell)_{xy}=\frac{1}{|S'(E_x)||S'(E_{\tau_1})||S'(E_{\tau_2})|\cdots|S'(E_{\tau_{\ell-1}})|}.
\end{align*}
By Lemma~\ref{lem:sp}(c), there exists $z\in E_x$ with $S(z)\in E_{\tau_1},S^2(z)\in E_{\tau_2},\ldots, S^{\ell}(z)\in E_y$, so 
\[
(P_n^\ell)_{xy}=\frac{1}{|S'(z)S'(S(z))\cdots S'(S^{\ell-1}(z))|}=\frac{1}{|(S^\ell)'(z)|}.
\]
\end{proof}

The following lemma demonstrates the sparseness of the matrices $P_n^\ell$ for times before ${\K}(n)$. Essentially, this is because for these times, the nonzero entries of the matrix $P_n^\ell$ are placed by drawing $S^\ell$ and an $n\times n$ grid in $[0,1]^2$, and then placing a nonzero entry in each position in the grid that $S^\ell$ passes through. As $n$ increases, the grid becomes finer and the graph of $S^\ell$ in $[0,1]^2$, which is one-dimensional, cannot pass through a very large fraction of the boxes.
\begin{lem}[number of nonzero entries]\label{lem:nonzero}
Assume (i)--(iv) and let $1\le \ell\le {\K}(n)+1$. Then the diagonal of $P_n^\ell$ contains at most $2 M_0 L_0^{\ell-1}$ nonzero entries, and in total $P_n^\ell$ has at most $n\cdot s_\mathrm{max}^\ell$ nonzero entries, where $s_\mathrm{max}$ is the maximum of the absolute values of the slopes in $S$.
\end{lem}
\begin{proof}
Pick an atom $E_x$. Since the maximum slope magnitude of $S$ is $s_\mathrm{max}$, the interval $S^\ell(E_x)$ has length at most $s_\mathrm{max}^\ell\cdot |E_x|$ and intersects at most $s_\mathrm{max}^\ell$ atoms $E_y$. Thus by Lemma~\ref{lem:ppowers} the $x$th row of $P_n^\ell$ has at most $s_\mathrm{max}^\ell$ nonzero entries, so in total $P_n^\ell$ has at most $n\cdot s_\mathrm{max}^\ell$ nonzero entries.

Also by Lemma~\ref{lem:ppowers}, nonzero diagonal elements of $P_n^\ell$ occur exactly when $S^\ell(E_x)\cap E_x\ne\emptyset$. Let $Q\subset[0,1]\times[0,1]$ be the diagonal chain of squares $Q=\bigcup_{x=0}^{n-1} (\frac{x}{n},\frac{x+1}{n})\times(\frac{x}{n},\frac{x+1}{n})$, so that the nonzero diagonal elements $(P_n^\ell)_{xx}$ occur exactly when the graph of $S^\ell$ intersects the $x$th square $(\frac{x}{n},\frac{x+1}{n})\times (\frac{x}{n},\frac{x+1}{n})$ (Figure~\ref{fig:diagonal}).

\begin{figure}[!ht]
\centering
\begin{tikzpicture}[scale=.5]
\foreach \m in {3,4}{
\draw[fill=gray!30](\m,\m)--(\m+1,\m)--(\m+1,\m+1)--(\m,\m+1)--cycle;
}
\foreach \m in {1,2,3,4,5}{
\draw (\m,\m)--(\m+1,\m)--(\m+1,\m+1)--(\m,\m+1)--cycle;
}
\foreach \end in {0,7}{
\draw (\end,-1)--(\end,7);
\draw[dotted] (\end,-1)--(\end,-2);
\draw[dotted] (\end,7)--(\end,8);
}
\node at (6.5,6.5) {\reflectbox{$\ddots$}};
\node at (.5,.5) {\reflectbox{$\ddots$}};
\draw[decorate, decoration={brace,amplitude=5pt}] (7,-2.2)--(0,-2.2);
\node[below] at (3.5,-2.6) {$|I_i|=\frac{1}{M_0\cdot L_0^{\ell-1}}$};
\draw[decorate, decoration={brace,amplitude=3pt}] (2,.9)--(1,.9);
\node[below] at (1.5,.7) {$\frac{1}{n}$};
\draw (1.5,-1)--++(63.4349:9cm); % https://tex.stackexchange.com/questions/43894/draw-lines-by-specifying-angles
\end{tikzpicture}
\caption{A line of slope $2$ intersecting two (open) boxes in the diagonal $Q$.}\label{fig:diagonal}
\end{figure}
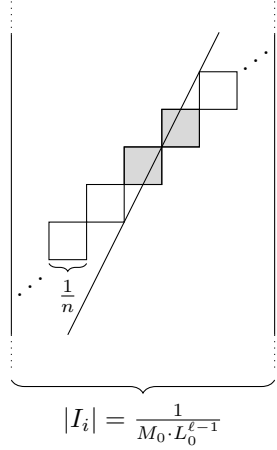

Choose an interval $I_i:=(\frac{i}{M_0L_0^{\ell-1}},\frac{i+1}{M_0L_0^{\ell-1}})$, an atom of the partition into $M_0L_0^{\ell-1}$ atoms. 
This is the coarsest partition for $S^\ell$ for which we can guarantee by Lemma~\ref{lem:sp}(a) that $S^\ell$ is linear on each atom. Since $\ell\le {\K}(n)+1$, the partition into $n$ atoms $E_x$ is a refinement of this one.
If the slope of $S^\ell$ is negative on $I_i$, then $S^\ell$ can intersect at most one box in the diagonal $Q$. If the slope of $S^\ell$ is positive and at least two on $I_i$, then it can intersect at most two boxes in $Q$ (see Figure~\ref{fig:diagonal}): 
Consider the slope one lines $t\pm\frac{1}{n}$ in $[0,1]^2$, which bound a parallelogram $R\supset Q$. Project the line segment $S^\ell(I_i)\cap R$ onto the $x$-axis. If $S^\ell$ on $I_i$ has slope $m>1$, then one can compute this projection is an interval of length $\le\frac{1}{n}\frac{2}{m-1}$. For $m\ge3$, this bound is $\le\frac{1}{n}$ so $S^\ell(I_i)\cap R$ can intersect at most two $\frac{1}{n}\times\frac{1}{n}$ boxes in $Q$. For $m=2$, the length can be $\frac{2}{n}$, but $S^\ell(I_i)\cap R$ can still only intersect at most two boxes in $Q$, by using that $S^\ell(\frac{j}{n})\in\frac{1}{n}\Z$ since $n$ is a multiple of $M_0L_0^{\ell-1}$. 

Then in total since there are $M_0L_0^{\ell-1}$ intervals $I_0,\ldots,I_{M_0L_0^{\ell-1}-1}$, there are at most $2M_0L_0^{\ell-1}$ nonzero entries on the diagonal of $P_n^\ell$.
\end{proof}

\begin{rmk}
Although the above argument works for slope $-1$, we do not allow slope $-1$ in $S$ since powers of $S$ could then have segments with slope $+1$.
\end{rmk}

\section{Proof of Theorem~\ref{thm:uweyl} pointwise Weyl law}\label{sec:weylproof}

In this section we prove Theorem~\ref{thm:uweyl} using a Fourier series approximation of the projection matrix, and properties of the Markov matrix $P_n$ and quantization $U_n$ to identify potentially bad coordinates $x$.
For notational convenience, we will use $n$ instead of $n_k$.
The idea is to approximate the spectral projection  matrix $P^{I(n)}:=\sum_{j:\theta^{(n,j)}\in I(n)}|\psi^{(n,j)}\rangle\langle\psi^{(n,j)}|$ using a Fourier series in powers of $U_n$,
\[
(P^{I(n)})_{xx}\approx \frac{|I(n)|}{2\pi}(1+o(1))+\sum_{|j|\ge1}a_j(U_n^j)_{xx}.
\]
The constant term in the Fourier series is already the leading term in the pointwise Weyl law. So it remains to show all the other terms are small, which can be done by splitting up the $|j|\ge1$ terms into three regions pictured as follows:
\begin{center}
\begin{tikzpicture}[xscale=1.6]
\def\th{.1cm} % axis line height
\def\teh{.6cm} % text height
\draw[->](0,0)--(6,0) node[right] {$|j|$};
\draw(0,\th)--(0,-\th) node[below] {$1$};
\draw(2,\th)--(2,-\th) node[below] {$r({\K}(n))$};
\draw(4,\th)--(4,-\th) node[below] {${\K}(n)$};
\node [align=left] at (5,\teh) {Fourier series\\convergence};
\node [align=left] at (3,\teh) {Semiclassical\\relation to $(P_n^j)_{xx}$};
\node [align=left] at (1,\teh) {Remove coords.\\
with short loops};
\end{tikzpicture}
\end{center}
For $|j|>{\K}(n)$, we have little knowledge of $(U_n^j)_{xx}$, so these terms are controlled simply by convergence of the Fourier series, or decay of the Fourier coefficients. For $|j|$ before what is essentially the Ehrenfest time ${\K}(n)$, we have good knowledge of $U_n^j$. We will split this region $1\le |j|\le {\K}(n)$ into two regions at a cut-off $r({\K}(n))$. Before the cut-off time $r({\K}(n))$, we remove all the coordinates $x$ that have short loops $x\to x$ of length $j$, since these coordinates can produce large Fourier series terms through the large entries $(U_n^j)_{xx}$. 
There are not many of these coordinates by Lemma~\ref{lem:nonzero}.
For $r({\K}(n))<|j|\le{\K}(n)$, we use the relationship between $U_n^j$ and $P_n^j$, specifically Lemma~\ref{lem:sp}(d), to control the size of $(U_n^j)_{xx}$.

Before starting the proof, we make some additional remarks about the proof and statement. 
\begin{rmk}\label{rmk:thm}
Let $r:\N\to\N$ be any function such that $r(m)<m$, like $r(m)=\lfloor m/2\rfloor$ or $\lfloor\log m\rfloor$. This is the cut-off function that determines which Fourier coefficients to examine for bad coordinates with short loops. For simplicity in the proof, one can take $r(\K(n))=\frac{K(n)}{2}$.
\begin{enumerate}[(i)]
\item To show the pointwise Weyl law \eqref{eqn:I}, we will show we can choose $G_n$ (not depending on $I(n)$) so that
$\#G_n\ge n-\frac{2M_0}{L_0-1} L_0^{r({\K}(n))}$, and for $x\in G_n$ that
\begin{multline}\label{eqn:close}
\left|\sum_{j:\theta^{(n,j)}\in I(n)}|\psi_x^{(n,j)}|^2-\frac{|I(n)|}{2\pi}\right| \\
\le \frac{|I(n)|}{2\pi}\left[2\pi|I(n)|^{-1}{\K}(n)^{-1}+(1+2\pi|I(n)|^{-1}{\K}(n)^{-1})\cdot 6\cdot2^{-r({\K}(n))/2}\right].
\end{multline}
Since $|I(n)|K(n)\to\infty$, the terms $|I(n)|^{-1}K(n)^{-1}$ on the right side are $o(1)$.

\item To ensure the entire right side of \eqref{eqn:close} is $o(|I(n)|)$ and that $\#G_n=n(1-o(|I(n)|))$, choose $r$ so that
\begin{align}\label{eqn:rconditions}
r({\K}(n))&\to\infty,\quad {\K}(n)-r({\K}(n))-\log_{L_0}\frac{1}{|I(n)|}\to\infty,\quad\text{as }{\K}(n)\to\infty.
\end{align} 
Since $|I(n)|{\K}(n)\to\infty$ by \eqref{eqn:km}, then eventually $\frac{1}{|I(n)|}<{\K}(n)$, so  condition \eqref{eqn:rconditions} is always met if
\begin{equation}
{\K}(n)-r({\K}(n))-\log_{L_0} {\K}(n)\to\infty.
\end{equation}
For example, $r({\K}(n))=\lfloor {\K}(n)/2\rfloor$ or $r({\K}(n))=\lfloor \log {\K}(n)\rfloor$ always satisfy the conditions on $r$.

\item We are interested in sequences of intervals $I(n)$ where $|I(n)|\to0$. For the proof, we will assume that $|I(n)|$ is bounded away from $2\pi$. If $|I(n)|$ is near $2\pi$,
apply the Theorem to the complement $I(n)^c$ or to a larger interval around $I(n)^c$ that satisfies \eqref{eqn:km}, to conclude 
\[
\left|\sum_{j:\theta^{(n,j)}\in I(n)}|\psi_x^{(n,j)}|^2-\frac{|I(n)|}{2\pi} \right|=\left|\sum_{j:\theta^{(n,j)}\in I(n)^c}|\psi_x^{(n,j)}|^2-\frac{|I(n)^c|}{2\pi}\right|\to0,
\]
as $|I(n)^c|\to0$.
\end{enumerate}
\end{rmk}

\subsection{Fourier series approximation}

Let $p_{I(n)}$ be the function on the unit circle in $\C$ defined by $p_{I(n)}(e^{it}):=\Chi_{I(n)}(t)$, so that $p_{I(n)}(U_n)$ is the projection
\[
P^{I(n)}:=p_{I(n)}(U_n)=\sum_{j:\theta^{(n,j)}\in I(n)}|\psi^{(n,j)}\rangle\langle\psi^{(n,j)}|.
\]
The sum $\sum_{j:\theta^{(n,j)}\in I(n)}|\psi_x^{(n,j)}|^2$ 
is the $(x,x)$ coordinate of the projection matrix $P^{I(n)}$.
To approximate $P^{I(n)}$ by a polynomial in powers of $U_n$, we approximate the indicator function $\Chi_{I(n)}$ by trigonometric polynomials. 

These particular polynomials are based on an entire function $B(z)$ introduced by Beurling, which satisfies $\operatorname{sgn}(x)\le B(x)$ for $x\in\R$, and $\int_\R (B(x)-\operatorname{sgn}(x))\,dx=1$. The function $B(z)$ also satisfies an extremal property; it minimizes the $L^1$ difference $\int_\R (f(x)-\operatorname{sgn}(x))\,dx$ over entire functions $f$ of exponential type\footnote{This is the growth condition for every $\varepsilon>0$, there is $A_\varepsilon$ so that $|f(z)|\le A_\varepsilon e^{(2\pi+\varepsilon)|z|}$ for all $z\in\C$.} $2\pi$ with $f(x)\ge\operatorname{sgn}(x)$ for $x\in\R$. By the Paley--Wiener theorem, exponential of type $2\pi$ means that the Fourier transform of $B(z)$ is supported in $[-2\pi,2\pi]$. 
Selberg later used this function $B(z)$ to produce majorants and minorants of the characteristic function $\Chi_I$ of an interval $I$, with compactly supported Fourier transform. 

\begin{thm}[Beurling--Selberg function]\label{thm:beurling-selberg}
Let $I\subset\R$ be a finite interval and $D>0$. Then there are functions $g^{(+)}_{I,D}$ and $g^{(-)}_{I,D}$ such that
\begin{enumerate}[(i)]
\item $g^{(-)}_{I,D}(x)\le\Chi_{I}(x)\le g^{(+)}_{I,D}(x)$ for all $x\in\R$.
\item The Fourier transforms $\hat{g^{(+)}_{I,D}}$ and $\hat{g^{(-)}_{I,D}}$ are compactly supported in $[-D,D]$.
\item $\int_\R |g^{(+)}_{I,D}(x)-\Chi_I(x)|\,dx=2\pi D^{-1}$ and $\int_\R|\Chi_I(x)-g^{(-)}_{I,D}(x)|\,dx=2\pi D^{-1}$.
\end{enumerate}
\end{thm}
The functions $g_{I,D}^{(\pm)}$ are not necessarily the minimizers of the $L^1$ difference from $\Chi_I$ (see \cite[p.4]{Vaaler}), but have still proven useful, including in Selberg's original number theoretic application. %to obtain a sharp form of the large sieve inequality.
For further references on Beurling and Selberg functions, see \cite[Chapter 45 \S20]{Selberg}, \cite{Montgomery}, or \cite{Vaaler}. 

For $I\subset\R/(2\pi\Z)$ with $|I|<2\pi$, to take $2\pi$-periodic functions, define 
\begin{align*}
G^{(+)}_{I,D}(x)&=\sum_{j\in\Z}g^{(+)}_{I,D}(x-2\pi j),\qquad G^{(-)}_{I,D}(x)=\sum_{j\in\Z}g^{(-)}_{I,D}(x-2\pi j),
\end{align*}
whose Fourier series coefficients then agree with the Fourier transform of $g^{(+)}_{I,D}$ or $g^{(-)}_{I,D}$ at integers,
\begin{align}
\hat{G^{(+)}_{I,D}}(k)&=\hat{g_{I,D}^{(+)}}(k),\qquad \hat{G^{(-)}_{I,D}}(k)=\hat{g_{I,D}^{(-)}}(k).
\end{align}
Thus also using property (iii), for $I\subset\R/(2\pi\Z)$ with $|I|<2\pi$,
\begin{align}\label{eqn:G}
G^{(\pm)}_{I,D}(x)=\frac{|I|\pm2\pi D^{-1}}{2\pi} + \sum_{\ell=1}^{\lfloor D\rfloor}\left(\hat{g^{(\pm)}_{I,D}}(\ell)e^{i\ell x}+\hat{g^{(\pm)}_{I,D}}(-\ell)e^{-i\ell x}\right).
\end{align}
These are trigonometric polynomials, sometimes called \emph{Selberg polynomials}, that approximate $\Chi_{I}$ well from above or below. 
Explicit expressions and plots for the Selberg polynomials can be found in \cite[\S1.2 pp.5--7]{Montgomery}.
The above equation \eqref{eqn:G} is the main approximation property we will need for later use.

\subsection{Projection matrix estimates}
Take $D={\K}(n)$, and define the functions on the unit circle in $\C$,
\begin{align*}
F^{(\pm)}_{I(n),{\K}(n)}(e^{it})&:= G^{(\pm)}_{I(n),{\K}(n)}(t).
\end{align*}
Recall we also defined $p_{I(n)}(e^{it})=\Chi_{I(n)}(t)$ and the projection $P^{I(n)}=p_{I(n)}(U_n)$, so that by the spectral theorem,
\begin{align}\label{eqn:sand2}
F_{I(n),{\K}(n)}^{(-)}(U_n)_{xx}\le (P^{I(n)})_{xx}\le F_{I(n),{\K}(n)}^{(+)}(U_n)_{xx}.
\end{align}
By \eqref{eqn:G} and the spectral theorem again, 
\begin{multline}
F^{(\pm)}_{I(n),{\K}(n)}(U_n) = \frac{|I(n)|}{2\pi}(1\pm 2\pi|I(n)|^{-1}{\K}(n)^{-1})\operatorname{Id}+\\
+\sum_{\ell=1}^{{\K}(n)}\left(\hat{g^{(\pm)}_{I(n),{\K}(n)}}(\ell)U_n^\ell
+\hat{g^{(\pm)}_{I(n),{\K}(n)}}(-\ell)U_n^{-\ell}\right).
\end{multline}
The identity term $\frac{|I(n)|}{2\pi}(1\pm 2\pi|I(n)|^{-1}{\K}(n)^{-1})\operatorname{Id}$ has the values we want already since $|I(n)|^{-1}{\K}(n)^{-1}\to0$ by \eqref{eqn:km}, so to show \eqref{eqn:close} we want to show the rest of the terms are small. Since
\begin{align}
|\hat{g^{(\pm)}_{I,D}}(\ell)|\le \frac{1}{2\pi}\int_\R|g^{(\pm)}_{I,D}(x)|\,dx\le \frac{1}{2\pi}(|I|+2\pi D^{-1}),
\end{align}
then for any $x,y\in\intbrr{n}$, the $(x,y)$ element of the non-identity terms can be bounded as
\begin{multline}\label{eqn:nonidbound}
\left|\sum_{\ell=1}^{{\K}(n)}\left(\hat{g^{(\pm)}_{I(n),{\K}(n)}}(\ell)(U_n^\ell)_{xy}+\hat{g^{(\pm)}_{I(n),{\K}(n)}}(-\ell)(U_n^{-\ell})_{xy}\right)\right| \le\\
\le \frac{|I(n)|}{2\pi}(1+2\pi|I(n)|^{-1}{\K}(n)^{-1})\sum_{\ell=1}^{{\K}(n)} (|(U_n^\ell)_{xy}|+|(U_n^\ell)_{yx}|).
\end{multline}

\subsection{Removing potentially bad points}\label{subsec:badpoints}
Here we use properties of $U_n$ and $P_n$ from Section~\ref{sec:plem} to remove coordinates $x$ where \eqref{eqn:nonidbound} may be large.
For $1\le \ell\le{\K}(n)+1$, by Lemma~\ref{lem:sp}(d) there is at most one path of length $\ell$ from a given $x$ to itself (or to another $y$), so
\begin{align*}
\left|(U_n^\ell)_{xx}\right|=\left|\sum_{\tau:x\xrightarrow{\ell}x}(U_n)_{\tau_0\tau_1}\cdots (U_n)_{\tau_{\ell-1}\tau_\ell}\right|
= \left|(U_n)_{x\tau_1}(U_n)_{\tau_1\tau_2}\cdots (U_n)_{\tau_{\ell-1}x}\right|=((P_n^\ell)_{xx})^{1/2}.
\end{align*}
Since all slopes of $S$ are at least 2 in absolute value, then all the slopes of $S^\ell$ are at least $2^\ell$ in absolute value, so by Lemma~\ref{lem:ppowers},
 $|(U_n^\ell)_{xx}|\le 2^{-\ell/2}$ for all $x\in\intbrr{n}$.
In order to make the sum $2\sum_{\ell=1}^{{\K}(n)}|(U_n^\ell)_{xx}|$ in \eqref{eqn:nonidbound} small then, we only need to be concerned with smaller $\ell$, since $|(U_n^\ell)_{xx}|$ decays exponentially in $\ell$. 
As we will see, by Lemma~\ref{lem:nonzero}, for small $\ell$, $(U_n^\ell)_{xx}=0$ for most coordinates $x$, so we can pick a cut-off for small $\ell$ and just throw out any coordinates $x$ where $(U_n^\ell)_{xx}\ne0$ below this cut-off. 

Let $r:\N\to\N$ satisfy $r(k)<k$ and \eqref{eqn:rconditions}; this will determine the cut-off for which $\ell$ are ``small''. 
Define the set of potentially bad coordinates as
\begin{equation}
B_n:=\{x\in\intbrr{n}:(U_n^\ell)_{xx}\ne 0\text{ for some }\ell\in\intbrr{1:\K(n)}\}.
\end{equation}
For $\ell\le{\K}(n)+1$, by Lemma~\ref{lem:nonzero}, the diagonal of $U_n^\ell$ contains at most $2\cdot M_0\cdot L_0^{\ell-1}$ nonzero entries, so there are not many bad points,
\begin{equation}\label{eqn:numbad}
\#B_n\le 2M_0\sum_{\ell=1}^{r({\K}(n))}L_0^{\ell-1} =\frac{2M_0}{L_0-1}(L_0^{r({\K}(n))}-1)=o(n|I(n)|),
\end{equation}
using assumption \eqref{eqn:rconditions} for the last equality. 
For $x\in G_n:=\intbrr{n}\setminus B_n$, then
\begin{align}\label{eqn:goodU}
\sum_{\ell=1}^{{\K}(n)}|(U_n^\ell)_{xx}|+|(U_n^{-\ell})_{xx}| &= 2\sum_{\ell=r({\K}(n))+1}^{{\K}(n)}|(U_n^\ell)_{xx}|\\ \nonumber&\le 2\sum_{\ell=r({\K}(n))+1}^\infty 2^{-\ell/2} =2(1+\sqrt{2})\cdot 2^{-r({\K}(n))/2}.
\end{align} 
Then for $x\in G_n$, 
\begin{align}
\Bigg|(P^{I(n)})_{xx}-\frac{|I(n)|}{2\pi}\Bigg| 
\nonumber &\le  \frac{|I(n)|}{2\pi}\left[2\pi|I(n)|^{-1}{\K}(n)^{-1}+(1+2\pi|I(n)|^{-1}{\K}(n)^{-1})\cdot 6\cdot2^{-r({\K}(n))/2}\right]\\
\nonumber &=o(|I(n)|),
\end{align}
since $|I(n)|{\K}(n)\to\infty$ by \eqref{eqn:km}. By \eqref{eqn:numbad}, $\#G_n\ge n(1-o(|I(n)|))$.
\qed

\begin{rmk}
The same method to estimate the diagonal entries of the projection matrix $P^{I(n)}$ can also be used to estimate the off-diagonal entries of $P^{I(n)}$. In this case, the constant term in the Fourier series expansion is zero, and then one can use similar arguments to show that the higher order terms are small. Alternatively, one can also obtain some bounds using that by the Weyl law, $\sum_{x,y\in\intbrr{n}}|(P^{I(n)})_{xy}|^2=\operatorname{tr}((P^{I(n)})^2)=\operatorname{tr}P^{I(n)}=\frac{n|I(n)|}{2\pi}(1+o(1))$.
\end{rmk}

%%%%%%%%%%%%%%%%%%%%%%%%%%%%%%%%%%%%%%%%%%%%%%%%%%%%%%%%%%
%%%%%%%%%%%%%%%%%%%%%%%%%%%%%%%%%%%%%%%%%%%%%%%%%%%%%%%%%%
\section{Quantum ergodicity in bins}\label{sec:qe}
%%%%%%%%%%%%%%%%%%%%%%%%%%%%%%%%%%%%%%%%%%%%%%%%%%%%%%%%%%
%%%%%%%%%%%%%%%%%%%%%%%%%%%%%%%%%%%%%%%%%%%%%%%%%%%%%%%%%%

In this section we prove Theorem~\ref{thm:qe} concerning quantum ergodicity in bins $\{j:\theta^{(n_k,j)}\in I(n_k)\}$, following the standard proof of quantum ergodicity that uses the Egorov property.

\begin{thm}[Egorov property, \cite{qgraphs}]\label{thm:egorov}
Suppose $S$ satisfies conditions (i)--(iv) and has a corresponding $n\times n$ unitary matrix $U_{n}$ with eigenvectors $(\psi^{(n,j)})_{j=1}^{n}$. Let $O_{n}(h)$ be the quantum observable corresponding to $h:[0,1]\to\C$. If $h$ is Lipschitz continuous on each image $S(E_x)$, and $n\in M_0L_0\Z$, then
\begin{equation}\label{eqn:egorov0}
\|U_{n}O_{n}(h)U_{n}^{-1}-O_{n}(h\circ S)\|\le\frac{1}{2}L_0^2M_0\cdot\frac{\displaystyle\max_{x\in\intbrr{n}}\|h\|_{\mathrm{Lip}(S(E_x))}}{n},
\end{equation}
where the norm on the left side is the operator norm, and the Lipschitz seminorm on the right is the Lipschitz constant of $h$ on $S(E_x)$, i.e. $\|h\|_{\operatorname{Lip}(S(E_x))}=\sup\limits_{z_1\ne z_2\in S(E_x)}\frac{|h(z_1)-h(z_2)|}{|z_1-z_2|}$.
\end{thm}
If $t\le{\K}(n)+1$, then by the same recursive argument as in Lemma~\ref{lem:sp}(a), $S^{t-1}$ is linear on each $S(E_x)$, so $h\circ S^{t-1}$ is Lipschitz on $S(E_x)$ with Lipschitz constant $\le\|h\|_\mathrm{Lip}L_0^{t-1}$, where here $\|h\|_\mathrm{Lip}=\sup_{z_1\ne z_2}\frac{|h(z_1)-h(z_2)|}{|z_1-z_2|}$ is the Lipschitz constant on the entire $[0,1]$. Then iterating \eqref{eqn:egorov0} $t$ times yields, 
\begin{align*}
\|U^t_nO_n(h)U_n^{-t}&-O_n(h\circ S^t)\|\\
&\le\sum_{r=1}^{t}\|U_n^{t-r}(U_nO_n(h\circ S^{r-1})U_n^{-1})U_n^{-(t-r)}-U_n^{t-r}O_n(h\circ S^r)U_n^{-(t-r)}\|\\
&\le \sum_{r=1}^t\|U_nO_n(h\circ S^{r-1})U_n^{-1}-O_n(h\circ S^r)\|\\
&\le \sum_{r=1}^t\frac{L_0^2M_0\|h\|_\mathrm{Lip}L_0^{r-1}}{2n}
\le \frac{C_S\|h\|_\mathrm{Lip}\cdot L_0^{t}}{n}.\numberthis\label{eqn:egorov}
\end{align*}
If say $t\le\frac{{\K}(n)}{2}$, then $L_0^{t}\ll n$, so the error bound is small, and the Egorov property \eqref{eqn:egorov} relates the quantum dynamics $U_n^tO_n(h)U_n^{-t}$ to the classical dynamics $h\circ S^t$ for $t$ well before the Ehrenfest time $T_E:={\K}(n)\lesssim \log n$.

\subsection{Proof of Theorem~\ref{thm:qe}}
Since $O_{n_k}(h)-(\int_0^1h)\cdot\mathrm{Id}=O_{n_k}(h-\int_0^1h)$, wlog assume $\int_0^1h=0$ and define the quantum variance for a fixed bin $I(n_k)$,
\begin{equation}
V_{n_k}:=\frac{1}{\#\{j:\theta^{(n_k,j)}\in I(n_k)\}}\sum_{j:\theta^{(n_k,j)}\in I(n_k)}\left|\langle \psi^{(n_k,j)},O_{n_k}(h)\psi^{(n_k,j)}\rangle\right|^2,
\end{equation}
which we will show tends to zero as $k\to\infty$. For a function $g:[0,1]\to\C$, define $[g]_T:=\frac{1}{T}\sum_{t=0}^{T-1}g\circ S^t$. Using that $\psi^{(n_k,j)}$ are eigenvectors of $U_{n_k}$, followed by the Egorov property and averaging over $t$ (stopping before $\frac{{\K}(n_k)}{2}$),
\begin{align*}
\langle \psi^{(n_k,j)},O_{n_k}(h)\psi^{(n_k,j)}\rangle &= \langle\psi^{(n_k,j)},(U_{n_k}^*)^tO_{n_k}(h)U_{n_k}^t\psi^{(n_k,j)}\rangle \\
&=\langle \psi^{(n_k,j)},O_{n_k}(h\circ S^t)\psi^{(n_k,j)}\rangle+\mathcal{O}\biggp{\frac{\|h\|_\mathrm{Lip}\cdot L_0^{t}}{n_k}}\\
&=\langle \psi^{(n_k,j)},O_{n_k}([h]_T)\psi^{(n_k,j)}\rangle+\mathcal{O}\biggp{\frac{\|h\|_\mathrm{Lip}\cdot L_0^{T}}{Tn_k}}.
\end{align*}
Then using the above followed by Cauchy--Schwarz,
\begin{align*}
|\langle \psi^{(n_k,j)},O_{n_k}(h)\psi^{(n_k,j)}\rangle|^2 &\le |\langle \psi^{(n_k,j)},O_{n_k}([h]_T)\psi^{(n_k,j)}\rangle|^2 + \mathcal{O}_h\biggp{\frac{L_0^{T}}{Tn_k}}\\
&\le \langle\psi^{(n_k,j)},O_{n_k}([h]^*_T)O_{n_k}([h]_T)\psi^{(n_k,j)}\rangle+\mathcal{O}_h\biggp{\frac{L_0^{T}}{Tn_k}},\numberthis
\end{align*}
where $\mathcal{O}_h(\cdot)$ indicates the constant in the big-$\mathcal{O}$ notation may depend on the function $h$.
For this quantization method, just a supremum norm bound on the integrand shows
\begin{align*}
|O_{n_k}(ab)_{xx}-O_{n_k}(a)_{xx}O_{n_k}(b)_{xx}| &= n\left|\int_{E_x}\left(a-\frac{1}{|E_x|}\int_{E_x}a\right)\left(b-\frac{1}{|E_x|}\int_{E_x}b\right)\right|\\
&\le \frac{\|a\|_{\mathrm{Lip}(E_x)}\|b\|_{\mathrm{Lip}(E_x)}}{n^2},
\end{align*}
so that
\begin{equation}
\|O_{n_k}(ab)-O_{n_k}(a)O_{n_k}(b)\|\le\max_{x\in\intbrr{1:n_k}}\frac{\|a\|_{\mathrm{Lip}(E_x)}\|b\|_{\mathrm{Lip}(E_x)}}{n_k^2}.
\end{equation}
Taking $T=\lfloor\frac{{\K}(n_k)}{2}\rfloor$, then for $t\le T$, $S^t$ is linear on every $E_x$ so that
$\|\frac{1}{T}\sum_{t=0}^{T-1}h\circ S^t\|_{\mathrm{Lip}(E_x)}\le  \frac{1}{T}\sum_{t=0}^{T-1}\|h\|_\mathrm{Lip}s_\mathrm{max}^t=\mathcal{O}_h(\frac{s_\mathrm{max}^T}{T})$, where $s_\mathrm{max}$ is the maximum absolute value of the slopes of $S$. Then
\begin{equation}\label{eqn:ipbound}
|\langle\psi^{(n_k,j)},O_{n_k}(h)\psi^{(n_k,j)}\rangle|^2\le \langle\psi^{(n_k,j)},O_{n_k}(|[h]_T|^2)\psi^{(n_k,j)}\rangle +\mathcal{O}_h\biggp{\frac{L_0^T}{Tn_k}}.
\end{equation}
Applying \eqref{eqn:ipbound} and the Weyl law analogue Corollary~\ref{cor:weyl}, followed by the pointwise Weyl analogue Theorem~\ref{thm:uweyl}, yields the bounds on the quantum variance,
\begin{align*}
&\frac{1}{\#\{j:\theta^{(n_k,j)}\in I(n_k)\}}\sum_{j:\theta^{(n_k,j)}\in I(n_k)}\left|\langle \psi^{(n_k,j)},O_{n_k}(h)\psi^{(n_k,j)}\rangle\right|^2 \\
&\qquad\le \frac{2\pi}{n_k|I(n_k)|(1+o(1))}\sum_{j:\theta^{(n_k,j)}\in I(n_k)}\langle\psi^{(n_k,j)},O_{n_k}(|[h]_T|^2)\psi^{(n_k,j)}\rangle+\mathcal{O}_h\biggp{\frac{L_0^{T}}{Tn_k}}\\
&\qquad\le\frac{2\pi(1+o(1))}{n_k|I(n_k)|}\left(\sum_{x\in G_{n_k}}\sum_{j:\theta^{(n_k,j)}\in I(n_k)}|\psi^{(n_k,j)}_x|^2O_{n_k}(|[h]_T|^2)_{xx}+\sum_{x\in B_{n_k}}\|h\|_\infty^2\right)+o(1)\\
&\qquad\le (1+o(1))\cdot \int_0^1\left|\frac{1}{T}\sum_{t=0}^{T-1}h(S^t(y))\right|^2\,dy+\frac{C\cdot L_0^{r({\K})}\|h\|_\infty^2}{n_k|I(n_k)|}+o(1)\xrightarrow{k\to\infty}0,
\end{align*}
using the $L^2$ ergodic theorem in the last line as $T=\lfloor\frac{{\K}(n_k)}{2}\rfloor\to\infty$.
\qed

The passage from decay of the quantum variance \eqref{eqn:qvariance} to the density one statement is by the usual method (for details see for example Theorem 15.5 in the textbook \cite{zworski}). To start, by Chebyshev--Markov with $\varepsilon=V_{n_k}^{1/4}$, Theorem~\ref{thm:qe} implies for a single Lipschitz function $h$, there is the sequence of sets $\Lambda_{n_k}(h)\subseteq\{j:\theta^{(n_k,j)}\in I(n_k)\}$ with 
\begin{equation}\label{eqn:density1}
\frac{\#\Lambda_{n_k}(h)}{\#\{j:\theta^{(n_k,j)}\in I(n_k)\}}\to1,
\end{equation}
such that for all sequences $(j_{n_k})_k$ with $j_{n_k}\in\Lambda_{n_k}(h)$, 
\begin{equation}\label{eqn:qe2}
\lim_{k\to\infty}\langle \psi^{(n_k,j_{n_k})},O_{n_k}(h)\psi^{(n_k,j_{n_k})}\rangle=\int_0^1h(x)\,dx.
\end{equation}
For a countable set of Lipschitz functions $(h_\ell)_\ell$, since finite intersections of sets $\Lambda_{n_k}$ satisfying \eqref{eqn:density1} also satisfy \eqref{eqn:density1}, we can assume $\Lambda_{n_k}(h_{\ell+1})\subseteq\Lambda_{n_k}(h_{\ell})$ for all $n_k$. Then for each $h_\ell$, let $N(\ell)>0$ be large enough so that for $n_k\ge N(\ell)$,
\begin{equation}
\frac{\#\Lambda_{n_k}(h_\ell)}{\#\{j:\theta^{(n_k,j)}\in I(n_k)\}}\ge 1-\frac{1}{\ell}.
\end{equation}
Take $N(\ell)$ increasing in $\ell$ and let $\Lambda_{n_k}^\infty:=\Lambda_{n_k}(h_\ell)$ for $N(\ell)\le n_k<N(\ell+1)$, so
\eqref{eqn:qe2}
 holds for sequences in $\Lambda^\infty_{n_k}$ and $h_\ell$ in the countable set. 
Then take $(h_\ell)_\ell$ to be a countable set of Lipschitz functions that are dense in $(C([0,1]),\|\cdot\|_\infty)$, so that for any $h\in C([0,1])$,
\begin{multline*}
\left|\langle \psi^{(n_k,j_{n_k})},O_{n_k}(h)\psi^{(n_k,j_{n_k})}\rangle-\int_0^1 h\right| \\
\le \left|\langle \psi^{(n_k,j_{n_k})},O_{n_k}(h-h_\ell)\psi^{(n_k,j_{n_k})}\rangle\right|+\\
+\left|\langle\psi^{(n_k,j_{n_k})},O_{n_k}(h_\ell)\psi^{(n_k,j_{n_k})}\rangle-\int_0^1 h_\ell\right|
+\left|\int_0^1(h_\ell-h)\right|.
\end{multline*}
The terms on the right side are bounded by $\|h-h_\ell\|_\infty$ or are $o(1)$ as $k\to\infty$.

%%%%%%%%%%%%%%%%%%%%%%%%%%%%%%%%%%%%%%%%%%%%%%%%%%%%%%%%%%
%%%%%%%%%%%%%%%%%%%%%%%%%%%%%%%%%%%%%%%%%%%%%%%%%%%%%%%%%%
\section{Random Gaussian eigenvectors}\label{sec:gaussian}
%%%%%%%%%%%%%%%%%%%%%%%%%%%%%%%%%%%%%%%%%%%%%%%%%%%%%%%%%%
%%%%%%%%%%%%%%%%%%%%%%%%%%%%%%%%%%%%%%%%%%%%%%%%%%%%%%%%%%

In this section we prove Theorems~\ref{thm:r-eigenvectors} and \ref{cor:gvectors} on random unitary rotations of bins of eigenvectors. 
We are interested in the value statistics of a random unit vector $v$ chosen from $V=\operatorname{span}\{\psi^{(n,j)}:\theta^{(n,j)}\in I(n)\}$. These coordinate values are just the one-dimensional projections of $v$ onto the standard basis vectors $\{e_1,\ldots,e_n\}$.
The behavior of such low-dimensional projections of high-dimensional vectors has been well-studied since the 1970s for its applications in analyzing large data sets; see for example the survey \cite{huber} for an overview of the early history and motivation of ``projection pursuit'' methods. The marginals of high-dimensional random vectors are often known to look approximately Gaussian, with precise conditions first proved by Diaconis and Freedman in \cite{DiaconisFreedman}. We will apply a quantitative version due to Meckes \cite{Meckes} and Chatterjee and Meckes \cite{ChatterjeeMeckes} to show the desired Gaussian coordinate value behavior for an orthonormal basis of such randomly chosen vectors.

\subsection{Random projections and bases}

First we start with Theorem~\ref{thm:r-eigenvectors}, which concerns the coordinate values of a single random unit vector in $V=\operatorname{span}\{\psi^{(n,j)}:\theta^{(n,j)}\in I(n)\}$. Let $M_V$ be the $n\times(\operatorname{dim}V)$ matrix whose columns are the basis vectors $(\psi^{(n,j)})$ of $V$. Then $P^{I(n)}:=M_VM_V^*$ is the projection onto this space, and a random unit vector $\phi$ in the span is chosen according to $\phi\sim g/\|g\|_2\sim M_Vu$, where $g\sim N_\C(0,P^{I(n)})$ and $u\sim\operatorname{Unif}(\mathbb{S}_\C^{\operatorname{dim}V-1})$.
The coordinates $\phi_1=\langle\phi,e_1\rangle,\phi_2,\ldots,\phi_n$ are
\[
\langle u,M_V^*e_1\rangle,\langle u,M_V^*e_2\rangle,\ldots,\langle u,M_V^*e_n\rangle
\]
which is a 1-dimensional projection in the direction $u\in\C^{\operatorname{dim}V}$ of the data set $\{M_V^*e_1,\ldots,M_V^*e_n\}\subset\C^{\operatorname{dim}V}$. Since $\sum_{x=1}^n|\phi_x|^2=1$, we will use the scaled data set $\sqrt{n}\{M_V^*e_1,\ldots,M_V^*e_n\}$. 

The following convergence result, which will prove Theorem~\ref{thm:r-eigenvectors}, is as follows. It is a quantitative version adapted to our case from \cite{Meckes}, also using \cite{ChatterjeeMeckes}, of the theorem from \cite{DiaconisFreedman}.
\begin{thm}[Complex projection version of Theorem 2 in \cite{Meckes}]\label{thm:complex-projection}
Let $P^{(\nu)}$ be an $n\times n$ self-adjoint projection matrix onto a $d$-dimensional subspace $V^{(\nu)}$ of $\C^n$, and suppose
\begin{align}
\sum_{x=1}^n\left|\|P^{(\nu)}e_x\|_2^2-\frac{d}{n}\right|&\le A.
\end{align}
Let $\omega=(\omega_1,\ldots,\omega_n)$ be chosen uniformly at random from the $(d-1)$-dimensional sphere $\mathbb{S}(V^{(\nu)}):=\{v\in V^{(\nu)}:\|v\|=1\}$, and define the empirical distribution $\tilde{\mu}_\omega^{(\nu)}$ of the coordinates of $\omega$ scaled by $\sqrt{n}$,
\begin{equation*}
\tilde{\mu}_\omega^{(\nu)} := \frac{1}{n}\sum_{x=1}^n\delta_{\sqrt{n}\omega_x}.
\end{equation*}
There are absolute numerical constants $C,c>0$ so that for $f:\C\to\C$ bounded $L$-Lipschitz and $\varepsilon>\frac{2L(A+3)}{d-1}$,
\begin{equation}\label{eqn:expbound}
\mathbb{P}\left[\left|\int f(x)\,d\mu^{(\nu)}_\omega(x)-\mathbb{E}f( Z)\right|>\varepsilon\right]\le C\exp(-c\varepsilon^2d/L^2),
\end{equation}
where $Z\sim N_\C(0,1)$.
\end{thm}
\begin{rmk}
We will consider $d=d(n_k)\to\infty$ as $n_k\to\infty$. As will be shown by \eqref{eqn:A} and \eqref{eqn:2Abound}, the pointwise Weyl law analogue Theorem~\ref{thm:uweyl} shows $A=o(d)$ with $d=\frac{n_k|I(n_k)|}{2\pi}(1+o(1))$, so the above theorem proves Theorem~\ref{thm:r-eigenvectors}.
\end{rmk}
The proof outline of Theorem~\ref{thm:complex-projection}, which will be a consequence of results from \cite{Meckes} and \cite{ChatterjeeMeckes}, is given in Appendix~\ref{sec:details}.

Theorem~\ref{thm:complex-projection} provides a bound for the probability that a single randomly chosen vector does not look Gaussian. Because the quantitative bound \eqref{eqn:expbound} decays quickly, a simple union bound gives a bound on finding an entire orthonormal basis that looks Gaussian (Corollary~\ref{cor:randomGaussian} below). This family of random orthonormal bases will then be used  to construct the unitary matrices $V_{n_k}(\beta^{[n_k]})$ in Theorem~\ref{cor:gvectors}.

\begin{cor}[Random Gaussian basis]\label{cor:randomGaussian}
Let $\C^n=V^{[1]}\oplus\cdots\oplus V^{[\kappa]}$, and let $P^{[\ell]}$ be the orthogonal projection onto the subspace $V^{[\ell]}$. Suppose there is $A$ and $d_1,\ldots,d_\kappa\in\R_+$ so that
\begin{equation}
\sum_{x=1}^n\left|\|P^{[\ell]}e_x\|_2^2-\frac{d_\ell}{n}\right|\le A,\quad \forall \ell\in\intbrr{\kappa}.
\end{equation}
Choose a random orthonormal basis $(\phi^{[j]})_{j=1}^n$ for $\C^n$ by choosing a random orthonormal basis from each $V^{[\ell]}$ (according to Haar measure) and embedding it by inclusion in $\C^n$. For each $j\in\intbrr{n}$, let 
\[
\mu^{[j]}:=\frac{1}{n}\sum_{x=1}^n\delta_{\sqrt{n}\phi^{[j]}_x},
\]
the empirical distribution for the $j$th basis vector's coordinates. There are absolute numerical constants $C,c>0$ so that for any $f:\C\to\C$ bounded $L$-Lipschitz and $\varepsilon>\frac{2L(2A+3)}{(\min d_\ell-A)-1}$,
\begin{equation}
\mathbb{P}\left[\max_{j\in\intbrr{n}}\left|\int f(x)\,d\mu^{[j]}(x)-\mathbb{E}f( Z)\right|>\varepsilon\right]\le Cn\exp\left(-\frac{c\varepsilon^2(\min d_\ell-A)}{L^2}\right),
\end{equation}
where $Z\sim N_\C(0,1)$.
\end{cor}
\begin{proof}
Recall that $\|P^{[\ell]}e_x\|_2^2=P^{[\ell]}_{xx}$, since $P^{[\ell]}$ is orthogonal projection. The numbers $d_\ell$ need not be the dimensions of $V^{[\ell]}$, but since
\begin{align}
\left|\frac{\operatorname{dim} V^{[\ell]}}{n}-\frac{d_\ell}{n}\right|&=\left|\frac{1}{n}\sum_{x=1}^n\Big(\|P^{[\ell]}e_x\|_2^2-\frac{d_\ell}{n}\Big)\right|\le\frac{1}{n}A,
\end{align}
then 
\begin{equation}\label{eqn:2Abound}
\sum_{x=1}^n\left|\|P^{[\ell]}e_x\|_2^2-\frac{\dim V^{[\ell]}}{n}\right|\le 2A.
\end{equation}
\begin{sloppypar}
Then Theorem~\ref{thm:complex-projection} implies for $\varepsilon>\frac{2L(2A+3)}{\min d_\ell-A-1}\ge\frac{2L(2A+3)}{\min\dim V^{[j]}-1}$, that 
$R_\ell(f):=\{\omega\in \mathbb{S}(V^{[\ell]}):\left|\int f(x)\,d\mu_\omega(x)-\mathbb{E}f(Z)\right|>\varepsilon\}$ has small measure $\le C\exp(-c\varepsilon^2\dim V^{[\ell]}/L^2)$. By Lemma~\ref{lem:onb} below, a random orthonormal basis for $V^{[\ell]}$ avoids $R_\ell(f)$ with probability at least $1-\dim V^{[\ell]}\cdot C\exp(-c\varepsilon^2\dim V^{[\ell]}/L^2)$. Thus letting $I_\ell\subset\intbrr{n}$ be the set of indices $j$ corresponding to $V^{[\ell]}$,
\begin{align*}
\mathbb{P}\Big[\max_{j\in\intbrr{n}}\Big|\int f(x)\,d\mu^{[j]}(x)&-\mathbb{E}f( Z)\Big|>\varepsilon\Big] \\
&\le \sum_{\ell=1}^\kappa\mathbb{P}\left[\left|\int f(x)\,d\mu^{[j]}(x)-\mathbb{E}f(Z)\right|>\varepsilon\text{ for some }j\in I_\ell\right]\\
&\le \sum_{\ell=1}^\kappa \operatorname{dim}V^{[\ell]}\cdot C\exp\left(-c\varepsilon^2\operatorname{dim}V^{[\ell]}/L^2\right)\\
&\le Cn\exp\left(-c\varepsilon^2\min \operatorname{dim}V^{[\ell]}/L^2\right)\le Cn\exp\left(-c\varepsilon^2(\min d_\ell-A)/L^2\right).
\end{align*}
\end{sloppypar}
\end{proof}

\begin{lem}[union bound for random ONB]\label{lem:onb}
Let $B\subset\mathbb{S}_\C^{d-1}$ and let $\sigma$ be the uniform surface measure on $\mathbb{S}_\C^{d-1}$ normalized so $\sigma(\mathbb{S}_{\C}^{d-1})=1$. Then a random orthonormal basis of $\C^d$ (chosen uniformly from Haar measure on the unitary group $U(d)$) avoids $B$ with probability at least $1-d\sigma(B)$.
\end{lem}
\begin{proof}
Let $\mu$ be normalized Haar measure on $U(d)$. Then for any $x\in\mathbb{S}_\C^{d-1}$, $\sigma(A)=\mu(g\in U(d):g(x)\in A)$. By union bound, letting $\{e_j\}$ be the standard basis,
\begin{align*}
\mu(\{g\in U(d):g(e_j)\in B\text{ for some }j\in\intbrr{1:d}\})&\le d\cdot \mu(\{g\in U(d):g(e_1)\in B\})\\
&= d\cdot\sigma(B),
\end{align*}
so $\mu(\{g\in U(d):\forall j\in\intbrr{1:d},\;g(e_j)\not\in B\})\ge 1-d\sigma(B)$.
\end{proof}

\subsection{Proof of Theorem~\ref{cor:gvectors}}\label{subsec:proof-gvectors}
Choose $\kappa(n_k)\in\N$ so that if we
divide $[0,2\pi]$ up into $\kappa(n_k)$ equal sized intervals $I_1(n_k),\ldots,I_{\kappa(n_k)}(n_k)$, then equation \eqref{eqn:km}, i.e. $|I(n_k)|{\K}(n_k)\to\infty$ as $k\to\infty$, holds for $|I(n_k)|=\frac{2\pi}{\kappa(n_k)}$. 
Let $\psi^{(n_k,j)}$ be the $j$th eigenvector of $U_{n_k}$.
Construct $V_{n_k}(\beta^{[n_k]})$ by taking a random unitary rotation (according to Haar measure) of the eigenvectors $\{\psi^{(n_k,j)}:\theta^{(n_k,j)}\in I_\ell(n_k)\}$, independently for each interval $I_\ell(n_k)$. 
The parameter $\beta^{[n_k]}$ represents these choices of random rotations of the eigenvectors within the intervals. 
Then perturb any degenerate eigenvalues to be simple by shifting them a small amount, while still keeping them in the same bin.
Denote the resulting eigenvectors of ${V_{n_k}}(\beta^{[n_k]})$  by $\phi^{[n_k,j]}_{(\beta)}$. In what follows the constant $C$ is a numerical constant that may change from line to line.
\begin{enumerate}[(a),leftmargin=.6cm,itemindent=0cm,labelwidth=\itemindent]

\item[(a)] Let $\tilde{U}_{n_k}$ be the perturbation of $U_{n_k}$ obtained by reassigning all eigenvalues in the same bin $I_\ell(n_k)$ to a single value $e^{i\Theta_\ell}$ in the bin. Then 
\[
\|(U_{n_k}-\tilde{U}_{n_k})v\|_2^2=\sum_{j=1}^n|e^{i\theta^{(j)}}-e^{i\Theta_\ell(j)}|^2|\langle\psi^{(j)},v\rangle|^2\le C\frac{(2\pi)^2}{\kappa(n_k)^2}\|v\|_2^2,
\]
since the reassigned eigenvalues are still in the same bin. Also, $\|\tilde{U}_{n_k}-V_{n_k}(\beta^{[n_k]})\|\le C\frac{2\pi}{\kappa(n_k)}$ by the same computation, since $\tilde{U}_{n_k}$ has degenerate eigenspaces that can be rotated to match the eigenvectors of $V_{n_k}(\beta^{[n_k]})$. Thus for any random $V_{n_k}(\beta^{[n_k]})$, $\|U_{n_k}-V_{n_k}(\beta^{[n_k]})\|\le C\frac{2\pi}{\kappa(n_k)}=o(1)$.
The Egorov property for $U_{n_k}$, Theorem~\ref{thm:egorov}, then implies the weaker Egorov property for $V_{n_k}(\beta^{[n_k]})$, since if $A$ and $B$ are unitary, then
\begin{align*}
\|AMA^{-1}-BMB^{-1}\| &=\|(A-B)MA^{-1}+BM(A^{-1}-B^{-1})
\le 2\|A-B\|\|M\|,
\end{align*}
and this also holds if we replace $M$ with $M-c\cdot\mathrm{Id}$ for any $c\in\C$ like $c=\int_0^1h$ or $c=h(0)$.

\item[(b)] 

To show Gaussian behavior, we first show there is $\varepsilon(n_k)\to0$ so that for any bounded Lipschitz $f:\C\to\C$, as $k\to\infty$,
\begin{equation}\label{eqn:gaussian}
\mathbb{P}\left[\max_{j\in\intbrr{1:n_k}}\left|\int f(x)\,d\mu^{[n_k,j]}_\beta(x)-\mathbb{E}f(Z)\right|>\|f\|_\mathrm{Lip}\varepsilon(n_k)\right]\le Cn_k\exp(-cn_k^{1/2}|I(n_k)|^{1/2}),
\end{equation}
where $Z\sim N_\C(0,1)$. 
A density argument followed by tightness will then complete the proof of (b).

To show \eqref{eqn:gaussian}, note that for any $W^{[\ell]}=\operatorname{span}\{\psi^{(n_k,j)}:\theta^{(j)}\in I_\ell(n_k)\}$ and $P^{[\ell]}$ the orthogonal projection onto $W^{[\ell]}$, the pointwise Weyl law Theorem~\ref{thm:uweyl} implies
\begin{align}\label{eqn:A}
\sum_{x=1}^n\left|\|P^{[\ell]}e_x\|_2^2-\frac{|I(n_k)|}{2\pi}\right| &\le \sum_{x\in G_{n_k}}\frac{|I(n_k)|}{2\pi}o(1) + \sum_{x\in B_{n_k}} 2 = o(n_k|I_\ell(n_k)|),
\end{align}
so the quantity $A$ in Corollary~\ref{cor:randomGaussian} can be taken to be $o(n_k|I_\ell(n_k)|)$.
Let $\mu^{[n_k,j]}_{(\beta)}$ be the coordinate distribution of the $j$th eigenvector $\phi^{[n_k,j]}_{(\beta)}$ of $V_{n_k}(\beta^{[n_k]})$.
Then applying Corollary~\ref{cor:randomGaussian} with all $d_\ell=\frac{n_k|I(n_k)|}{2\pi}$ and 
\begin{equation*}
\varepsilon(n_k)=\max\Big(\frac{4A+6}{(d_\ell-A)-1},\frac{1}{(n_k|I(n_k)|)^{1/4}}\Big)\to0,
\end{equation*}
this yields for $Z\sim N_\C(0,1)$,
\begin{align*}
\mathbb{P}\Big[\max_{j\in\intbrr{1:n_k}}\Big|\int_\C f(x)\,d\mu^{[n_k,j]}_{(\beta)}(x)&\mathbb{E}f( Z)\Big|>\|f\|_\mathrm{Lip}\varepsilon(n_k)\Big]\\
&\le Cn_k\exp\left(-c\varepsilon(n_k)^2(n_k|I(n_k)|-2\pi A)\right)\\
&\le Cn_k\exp\left(-cn_k^{1/2}|I(n_k)|^{1/2}\Big(1-\frac{2\pi A}{n_k|I(n_k)|}\Big)\right).\numberthis\label{eqn:expbound2}
\end{align*}

Now let $(f_\ell)_\ell$ be a countable set of Lipschitz functions with compact support that are dense in $C_c(\C)$, %in uniform norm
and set
\begin{align*}
\Pi_{n_k}=\Big\{V_{n_k}(\beta^{[n_k]}):\forall \ell\in\intbrr{1:n_k},{j\in\intbrr{1:n_k}},
\left|\int_\C f_\ell(x)\,d\mu^{[n_k,j]}_{(\beta)}-\mathbb{E}f_\ell(Z)\right|\le\|f_\ell\|_{\mathrm{Lip}} \varepsilon(n_k)\Big\}.
\end{align*}
Then
\begin{equation}
\mathbb{P}[\Pi_{n_k}^c]\le Cn_k^2\exp(-cn_k^{1/2}|I(n_k)|^{1/2})\to0,
\end{equation}
since by \eqref{eqn:km}, $n_k^{1/2}|I(n_k)|^{1/2}\gg 2\log n_k$.
For a sequence of matrices $(\tilde{V}_{n_k})_k$ with $\tilde{V}_{n_k}\in\Pi_{n_k}$, let $\tilde{\mu}^{[n_k,j]}$ be the scaled coordinate distribution of the $j$th eigenvector $\tilde{\phi}^{[n_k,j]}$ of $\tilde{V}_{n_k}$. By definition of $\Pi_{n_k}$, we know for any $f_\ell$ that $\int f_\ell\,d\mu^{[n_k,j_{n_k}]}\to\mathbb{E}f_\ell(Z)$ as $k\to\infty$, for any sequence $(j_{n_k})_k$ with $j_{n_k}\in\intbrr{1:n_k}$. Denseness of $(f_\ell)_\ell$ shows that this holds for all $f\in C_c(\C)$ as well.
Then $(\tilde{\mu}^{[n_k,j_{n_k}]})_k$ is tight, and with the vague convergence we get weak convergence of $\tilde{\mu}^{[n_k,j_{n_k}]}$ to $N_\C(0,1)$.

\item[(c)]
To show QUE, like in (b), we first show there is $\varepsilon(n_k)\to0$ so that for any bounded Lipschitz $h:[0,1]\to\C$, as $k\to\infty$,
\begin{multline}\label{eqn:que}
\mathbb{P}\left[\max_{j\in\intbrr{1:n_k}}\left|\langle\phi^{[n_k,j]}_{(\beta)},O_{n_k}(h)\phi^{[n_k,j]}_{(\beta)}\rangle - \int_0^1h(x)\,dx\right|>\|h\|_\infty\varepsilon(n_k)\right]
\le Cn_k\exp(-c n_k^{1/2}|I(n_k)|^{1/2}).
\end{multline}
This is done by the same argument presented in \cite{randomQUE} using the Hanson--Wright inequality \cite{hanson-wright}. After proving \eqref{eqn:que}, part (c) follows from density like in (b).

\begin{sloppypar}
For $W^{[\ell]}$ with dimension $d$, let $M_{W^{[\ell]}}$ be an $n\times d$ matrix whose $d$ columns form an orthonormal basis for $W^{[\ell]}$.
Then $\phi^{[n_k,j]}$ chosen randomly from $\mathbb{S}(W^{[\ell]})$ is distributed like $M_{W^{[\ell]}}u$ for $u\sim\operatorname{Unif}(\mathbb{S}_\C^{d-1})$, and
\begin{equation}
\langle \phi^{[n_k,j]},O_{n_k}(h)\phi^{[n_k,j]}\rangle_{\C^n}\sim \langle u,(M_{W^{[\ell]}}^*O_{n_k}(h)M_{W^{[\ell]}})u\rangle_{\C^d}.
\end{equation}
The Hanson--Wright inequality combined with subgaussian concentration on the norm $\|N_\C(0,I_d)\|_2$ shows that $\langle u,(M_{W^{[\ell]}}^*O_{n_k}(h)M_{W^{[\ell]}})u\rangle$ concentrates around its mean $\frac{1}{d}\operatorname{tr}(M_{W^{[\ell]}}^*O_{n_k}(h)M_{W^{[\ell]}})$ (see \cite[Theorem 4.1]{randomQUE} for details), which by the pointwise Wey law Theorem~\ref{thm:uweyl} is $\int_0^1h+R(n_k,h)$ with $|R(n_k,h)|\le R(n_k)\|h\|_\infty$, some $R(n_k)\to0$. In particular, for any $\varepsilon>2|R(n_k,h)|$,
\begin{equation}
\mathbb{P}\left[\left|\langle\phi^{[n_k,j]}_{(\beta)},O_{n_k}(h)\phi^{[n_k,j]}_{(\beta)}\rangle-\int_0^1h\right|>\varepsilon\right] \le C_1\exp\left(-C_2\min\left(\frac{\varepsilon^2}{4\|h\|_\infty^2},\frac{\varepsilon}{2\|h\|_\infty}\right) \cdot d\right).
\end{equation}
Then taking $\varepsilon(n_k)=\max\left(2R(n_k),\frac{1}{(n_k|I(n_k)|)^{1/4}}\right)$ and applying a union bound like with Lemma~\ref{lem:onb} yields \eqref{eqn:que}, using that eventually $\min(\varepsilon(n_k)^2,\varepsilon(n_k))=\varepsilon(n_k)^2$.
\end{sloppypar}

\begin{sloppypar}
Next, taking $(h_\ell)_\ell$ to be a countable dense set of Lipschitz functions in $C([0,1])$, let
\begin{multline}
\Gamma_{n_k}=\left\{V_{n_k}(\beta^{[n_k]}):\forall\ell\in\intbrr{1:n_k}, {j\in\intbrr{1:n_k}},\vphantom{\int_0^1}\right.\\
\left.\left|\langle\phi^{[n_k,j]}_{(\beta)},O_{n_k}(h_\ell)\phi^{[n_k,j]}_{(\beta)}\rangle-\int_0^1h_\ell(x)\,dx\right|
\le\|h_\ell\|_\infty\varepsilon(n_k)\right\}.
\end{multline}
Then 
$\mathbb{P}[\Gamma_{n_k}^c]\le C_1n_k^2\exp(-C_2n_k^{1/2}|I(n_k)|^{1/2})\to0$, and denseness shows that for sequences $(\tilde{V}_{n_k})_k$ with $\tilde{V}_{n_k}\in\Gamma_{n_k}$ and with eigenvectors denoted by $\tilde{\phi}^{[n_k,j]}$, that $\langle\tilde{\phi}^{[n_k,j_{n_k}]},O_{n_k}(h)\tilde{\phi}^{[n_k,j_{n_k}]}\rangle\to\int_0^1h$ for all $h\in C([0,1])$ as well.
\end{sloppypar}

\item[(d)] To make the spectrum simple, we simply perturbed any degenerate eigenvalues while keeping them in the same bin.

\item[(e)] This follows from $\|U_{n_k}-V_{n_k}(\beta^{[n_k]})\|\le C\cdot\frac{2\pi}{\kappa(n_k)}$, since for any matrices $U$ and $V$ with entries $|\cdot|\le1$, 
\[
\left||V_{xy}|^2-|U_{xy}|^2\right|\le 2\left||V_{xy}|-|U_{xy}|\right|\le2|V_{xy}-U_{xy}|\le2\|V-U\|.
\] 
\qed
\end{enumerate}

%%%%%%%%%%%%%%%%%%%%%%%%%%%%%%%%%%%%%%%%%%%%%%%%%%%%%%%%%%
%%%%%%%%%%%%%%%%%%%%%%%%%%%%%%%%%%%%%%%%%%%%%%%%%%%%%%%%%%
\section{The doubling map with any even $n$}\label{sec:doubling-even}
%%%%%%%%%%%%%%%%%%%%%%%%%%%%%%%%%%%%%%%%%%%%%%%%%%%%%%%%%%
%%%%%%%%%%%%%%%%%%%%%%%%%%%%%%%%%%%%%%%%%%%%%%%%%%%%%%%%%%

Recall if $S:[0,1]\to[0,1]$ is the doubling map $S(x)=2x\pmod{1}$, then for any $n\in2\N$, the $n\times n$ Markov matrix $P_n$ along with a specific quantization $U_n$ can be taken as in \eqref{eqn:doubling-pu}.
For general maps $S$, in Theorem~\ref{thm:uweyl}, we restricted to dimensions $n\in M_0\Z$ with ${\K}(n)\to\infty$. This ensured that enough powers of $P_n$ behaved nicely with the partitions (Lemma~\ref{lem:ppowers}). For the doubling map, $M_0=2$ and we can take all $n\in 2\Z$, not just those with the largest power of two dividing $n$ tending to infinity.
The statement  is as follows (note that the quantization $U_n$ does not have to be the real orthogonal one in  \eqref{eqn:doubling-pu}).

\begin{thm}[pointwise Weyl law analogue for the doubling map]\label{prop:dweyl}
For $n\in2\N$, let $P_n$ be the matrix in \eqref{eqn:doubling-pu}, and let $U_n$ satisfy $|(U_n)_{xy}|^2=(P_n)_{xy}$. Denote the  eigenvalues and eigenvectors of $U_n$ by $(e^{i \theta^{(n)}_j})_j$ and $(\psi^{(n,j)})_j$ respectively for $j\in\intbrr{n}$. Let $(I(n))$ be a sequence of intervals in $\R/(2\pi \Z)$ satisfying
\begin{equation}
{|I(n)| \log n}\to\infty,\quad\text{as }n\to\infty.
\end{equation} 
Then there is a sequence of subsets $G_n\subseteq \intbrr{n}$ with sizes $\#G_n=n(1-o(|I(n)|))$ so that for all $x\in G_n$,
\begin{equation}
\sum_{j:\theta^{(n,j)}\in I(n)}|\psi_x^{(n,j)}|^2=\frac{|I(n)|}{2\pi}(1+o(1)),
\end{equation}
where the error term $o(1)$ depends only on $n$, $|I(n)|$, and $\#G_n$, and is independent of $x\in G_n$. Additionally, $G_{n_k}$ can be chosen independent of $I(n_k)$ or $|I(n_k)|$.
\end{thm}

The proof is the same as Theorem~\ref{thm:uweyl}, except that Lemma~\ref{lem:nonzero}, which bounds the number of nonzero entries on the diagonal of $P_n^\ell$, is proved differently.
To analyze the matrix powers $P_n^\ell$, instead of viewing them in terms of $S^\ell$, we count paths of length $\ell$ in the directed graph associated with the Markov matrix $P_n$. The proof of Theorem~\ref{prop:dweyl} then follows from the following lemma and by replacing all instances of ${\K}(n)+1$ by $K:=\lfloor\log_2n\rfloor$ in the proof of Theorem~\ref{thm:uweyl}.

\begin{lem}[number of nonzero entries for the doubling map]\label{lem:doubling-nonzero}
For $n\in2\N$, let $P_n$ be as in \eqref{eqn:doubling-pu} and let $1\le\ell\le K$, where $K=\lfloor\log_2n\rfloor$. Consider the directed graph with $n$ nodes $1,2,\ldots,n$, whose adjacency matrix is $2P_n$. Then:
\begin{enumerate}[(i)]
\item For any coordinates $x,y$, there is at most one path of length $\ell$ from $x$ to $y$ in the graph.
\item The diagonal of $P_n^\ell$ contains at most $2\cdot 2^\ell$ nonzero entries.
\item In total, $P_n^\ell$ has exactly $n\cdot 2^\ell$ nonzero entries.
\end{enumerate}
\end{lem}
\begin{proof}
All possible paths starting at a node $x$ and of length $\ell$ can be represented as paths in a binary tree of height $\ell$ with root node (height 0) $x$. (Figure~\ref{fig:tree}.) The nodes $1,2,\ldots,n$ of the graph may be listed multiple times in the binary tree.
If we always put the descendant $2x-1$ on the left and put $2x$ on the right, then the list of nodes in each row of the tree will be consecutive increasing in $\Z/n\Z$. Thus if $\ell\le K:=\lfloor\log_2 n\rfloor$, the $\ell$th row of the tree will contain $2^\ell\le n$ nodes, so each label $y$ can appear at most once. This implies for any two nodes $x$ and $y$, there is at most one path of length $\ell$ from $x$ to $y$, proving part (i).

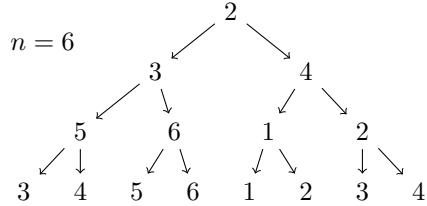
\begin{figure}[ht]
\centering
\begin{tikzpicture}[yscale=.8]
\node (0) at (0,0) {$2$};
\node (11) at (-1,-1) {$3$};
\node (12) at (1,-1) {$4$};
\foreach \i[evaluate={\c=int(6-Mod(-(\i+4),6));},evaluate={\d=int(6-Mod(-(\i+2),6));}] in {1,2,3,4}
{
\node (2\i) at (-3.25+1.25*\i,-2) {\c};
\node (3\i) at (-3.5+.75*\i,-3) {\d};
\node (4\i) at (-.5+.75*\i,-3) {\i};
}
\path[->]
	(0) edge (11)
	(11) edge (21)
	(11) edge (22)
	(12) edge (23)
	(12) edge (24)
	(21) edge (31)
	(22) edge (33)
	(23) edge (41)
	(23) edge (42)
	(24) edge (43)
	(24) edge (44);
\path[->,black]
	(0) edge node[above right] {} (12)
	(21) edge node[right] {} (32)
	(22) edge node[right] {} (34);
\node at (-2.5,-.5) {$n=6$};
\end{tikzpicture}
\caption{Start of a binary tree for $n=6$ ($K=2$). 
This tree describes all paths of length 3 that start at node 2.}\label{fig:tree}
\end{figure}

Applying part (i), the  total number of nonzero entries on the diagonal of $P_n^\ell$ is the total number of paths of length $\ell\le K$ with the same start and end point $x$.
Similarly, the total number of nonzero entries in $P_n^\ell$ is the total number of length $\ell$ paths from any $x$ to any $y$. 
The collection of all paths of length $\ell$ can be represented by the paths in a forest of $n$ binary trees each of depth $\ell$, one tree for each possible starting node $x\in\intbrr{n}$. (Figure~\ref{fig:forest}.) The $\ell$th row contains $2^\ell\cdot n$ numbers, showing part (iii).

\begin{figure}[!ht]
\centering
\begin{tikzpicture}[yscale=.8]
\foreach \x in {1,2}{
	\node at (1.5*\x,0) {$\x$};
	\draw (1.5*\x+.4,-1.5)--(1.5*\x,-.5)--(1.5*\x-.4,-1.5);
	\draw[yshift=-.15cm,decorate, decoration={brace,amplitude=5pt}] (1.5*\x+.4,-2)--(1.5*\x-.4,-2);
}
\node[below] at (1.5,-2.3) {$2^\ell$};
\node[below] at (3,-2.3) {$2^\ell$};
\node at (6,0) {$\ldots\,\ldots$};
\node[below] at (1.5,-1.5) {$[1,2,\ldots$};
\node[below] at (3,-1.67) {$\ldots$};
\node[below,xshift=.5cm] at (4.8,-1.5) {$n][1,2,\ldots][$};
\node at (6,-1) {$\ldots\,\ldots$};
% s1
\node at (3.9,0) {$\ldots$};
\node at (3.9,-1) {$\ldots$};
\draw[xshift=.3cm] (4.1,-1.5)--(4.5,-.5)--(4.9,-1.5);
\node[xshift=.3cm] at (4.5,0) {$s$};
% n stuff
\node[xshift=1.3cm] at (6,0) {$n$};
\draw[xshift=1.3cm] (6.4,-1.5)--(6,-.5)--(5.6,-1.5);
\node[below,xshift=1.3cm] at (6,-1.5) {$\ldots,n]$};
\draw[decorate,decoration={brace,amplitude=5pt},xshift=1.5cm] (6.4,-.5)--(6.4,-1.5) node[midway,right,xshift=.2cm] {$\ell$};
\end{tikzpicture}
\caption{All paths of length $\ell$ as paths in a forest.}\label{fig:forest}
\end{figure}
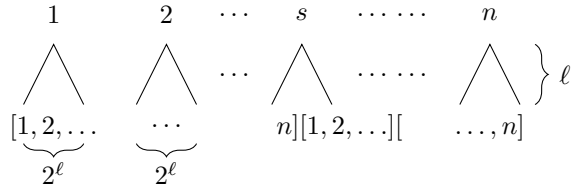

These $2^\ell\cdot n$ numbers at the bottom of the forest are $2^\ell$ copies of the sequence $(1,2,\ldots,n)$. 
To show (ii), we will show that for each copy $C_j$ of $(1,2,\ldots,n)$, there can be at most two paths with the same start and end point that end in this copy.

For $j=1,\ldots,2^\ell$, let $F(j)$ be the set of starting nodes at the top of the forest that have descendants in the $j$th copy $C_j$ of $(1,2,\ldots,n)$. (See Figure~\ref{fig:Fj}. Note the last node in $F(j)$ may overlap with the first node in $F(j+1)$.) Consider just the paths that start in $F(j)$ and end in $C_j$, and suppose there is a length $\ell$ path $x\to x$.  

\begin{figure}[!ht]
\begin{tikzpicture}
%%% binary trees 
%more tree options in https://latexdraw.com/draw-trees-in-tikz/
[
level 1/.style = {level distance = 7mm, sibling distance = 1cm},
level 2/.style = {level distance = 7mm, sibling distance = 5mm}
]
\node {1}
    child {node {1} child {node {1}} child {node {2}} edge from parent}
    child {node {2} child {node {3}} child {node {4}} edge from parent};
\node [xshift=2cm] {2}
    child {node {3} child {node {5}} child {node {6}} edge from parent}
    child {node {4} child {node {7}} child {node {8}} edge from parent};
\node [xshift=4cm] {3}
    child {node {5} child {node {9}} child {node {10}} edge from parent}
    child {node {6} child {node {1}} child {node {2}} edge from parent};
\node [xshift=6cm] {4}
    child {node {7} child {node {3}} child {node {4}} edge from parent}
    child {node {8} child {node {5}} child {node {6}} edge from parent};
\node [xshift=8cm] {5}
    child {node {9} child {node {7}} child {node {8}} edge from parent}
    child {node {10} child {node {9}} child {node {10}} edge from parent};
\node [xshift=12cm] {10}
    child {node {9} child {node {7}} child {node {8}} edge from parent}
    child {node {10} child {node {9}} child {node {10}} edge from parent};
%%%%%%% 
\node at (10,0) {$\cdots$};
\node[above] at (10,-1.5) {$\cdots$};
%%%%%%%% F(j)
\draw[dashed] (-.2,0.2)--(4.2,.2)--(4.2,-.2)--(-.2,-.2)--cycle;
\node[above] at (2,.2) {$F(1)$};
\draw[xshift=4cm,yshift=.5mm] (-.2,0.2)--(4.2,.2)--(4.2,-.2)--(-.2,-.2)--cycle;
\node[above,xshift=4cm] at (2,.2) {$F(2)$};
\draw[dashed,xshift=1.5cm] (10,.2)--(11,.2)--(11,-.2)--(10,-.2);
%%%%%%% C_j
\draw[dashed] (-1,-1.2)--(4,-1.2)--(4,-1.6)--(-1,-1.6)--cycle;
\node[below] at (1.8,-1.6) {$C_1$};
\draw[xshift=5.1cm] (-1,-1.2)--(4,-1.2)--(4,-1.6)--(-1,-1.6)--cycle;
\node[below,xshift=4.8cm] at (1.8,-1.6) {$C_2$};
\end{tikzpicture}
\caption{Paths of length $\ell=2$ for $n=10$. A starting set $F(2)$ and ending set $C_2$ are emphasized in solid outlines. In this case, the only loop starting in $F(2)$ and ending in $C_2$ is the one starting and ending at $4$.}\label{fig:Fj}
\end{figure}
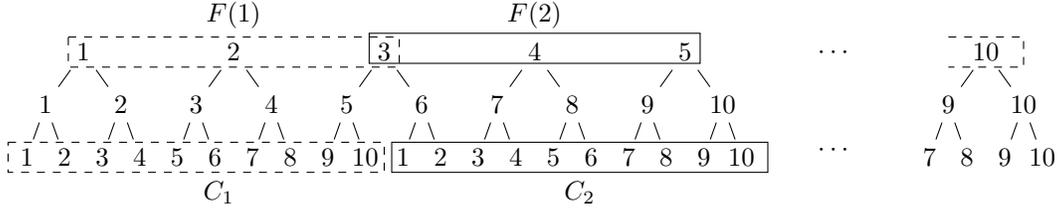

We claim that if there is another loop of length $\ell$ starting in $F(j)$ and ending in $C_j$, then this loop must begin either at $x-1$ or $x+1$ in $F(j)$. Morever, only at most one of $x-1$, $x+1$ can have the loop. (See Figure~\ref{fig:loops}.)
Let $L^\ell(x)$ be the left-most descendant of $x$ in $C_j$, and let $R^\ell(x)$ be the right-most descendant of $x$ in $C_j$.
\begin{enumerate}[(a)]
\item If $L^\ell(x)<x<R^\ell(x)$, then no other $y\in F(j)$ has a path $y\to y$. (Use $L^\ell(x+1)\ge x+2$ and $R^\ell(x-1)\le x-2$, and then continue for the rest of $F(j)$ using $L^\ell(y+1)\ge L^\ell(y)+2$ and $R^\ell(y-1)\le R^\ell(y)-2$.)
\item Similarly, if $L^\ell(x)=x$, then only also $x-1$ has a path $x-1\to x-1$.
\item If $R^\ell(x)=x$, then only also $x+1$ has a path $x+1\to x+1$.
\end{enumerate}

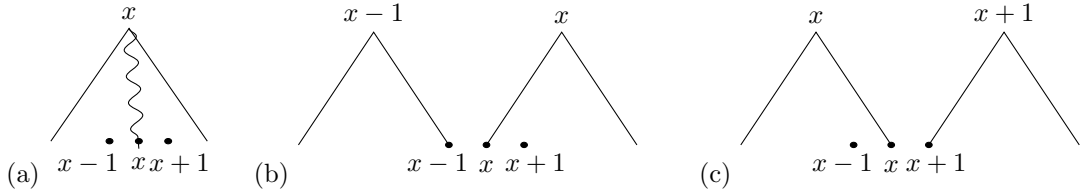
\begin{figure}[!ht]
(a)
\begin{tikzpicture}[xscale=1.3]
\draw (-.8,-1.5)--(0,0)--(.8,-1.5);
\node[above] at (0,0) {$x$};
\foreach \i in {-.3,0,.3}{
	\filldraw (\i+.1,-1.5) circle (1pt);
}
\node[below] at (.1,-1.6) {$x$};
\node[below,xshift=.15cm] at (.4,-1.55) {$x+1$};
\node[below left,xshift=-.15cm] at (.1,-1.55) {$x-1$};
\draw[decorate,decoration={snake}] (0,0)--(.1,-1.6);
\end{tikzpicture}
\quad(b)
\begin{tikzpicture}[xscale=1.25]
\foreach \x in {0,2}{
	\draw (\x-.8,-1.5)--(\x,0)--(\x+.8,-1.5);
}
\node[above] at (0,0) {$x-1$};
\node[above] at (2,0) {$x$};
\filldraw (1.2,-1.5) circle (1pt) node[below,yshift=-.09cm] {$x$};
\filldraw (1.6,-1.5) circle (1pt) node[below,xshift=.15cm] {$x+1$};
\filldraw (.8,-1.5) circle (1pt) node[below,xshift=-.15cm] {$x-1$};
\end{tikzpicture}
\qquad(c)
\begin{tikzpicture}[xscale=1.25]
\foreach \x in {0,2}{
	\draw (\x-.8,-1.5)--(\x,0)--(\x+.8,-1.5);
}
\node[above] at (0,0) {$x$};
\node[above] at (2,0) {$x+1$};
\filldraw (1.2,-1.5) circle (1pt) node[below] {\;\;$x+1$};
\filldraw (.4,-1.5) circle (1pt) node[below,xshift=-.15cm] {$x-1$};
\filldraw (.8,-1.5) circle (1pt) node[below,yshift=-.08cm] {$x$};
\end{tikzpicture}
\caption{The possible cases if there is a loop $x\to x$.}\label{fig:loops}
\end{figure}

Thus in total there are at most $2\cdot 2^\ell$ paths of length $\ell$ that start and end at the same point, proving part (ii).
\end{proof}

%%%%%%%%%%%%%%%%%%%%%%%%%%%%%%%%%%%%%%%%%%%%%
\section{The doubling map with \texorpdfstring{$n=2^K$}{n=2\^{}K}}\label{sec:doubling-2}

When $n=2^K$ for some $K\in\N$, the corresponding graphs from the doubling map are the de Bruijn graphs on two symbols. 
Orbits in these graphs have been studied in the context of quantum chaos in for example \cite{tanner2,leroux,GutkinOsipov,HarrisonHudgins}.
In these dimensions, the particular matrices $U_n$ from \eqref{eqn:doubling-pu}, despite coming from the doubling map, exhibit some behavior like that of integrable systems. Any choice of eigenbasis still satisfies the quantum ergodic theorem since the doubling map is ergodic, but the eigenvalues of $U_n$ in these dimensions are degenerate and evenly spaced in the unit circle. 
As a result of the degeneracy, we will be able to show that  random eigenbases look approximately Gaussian.
This will follow from properties of the spectral projection matrix of an eigenspace combined with the results on random projections used in Section~\ref{sec:gaussian}.

We start by showing the eigenvalues of $U_n$ are $4K$th roots of $1$ if $K$ is even, and $4K$th roots of $-1$ if $K$ is odd.

\begin{prop}[Repeating powers of $U_n$] 
\label{prop:u2k}
Let $U_n$ be as in \eqref{eqn:doubling-pu} with $n=2^K$. Then
\begin{enumerate}[(a)]
\item $U_n^{4K}=(-1)^KI$.
\item $U_n^r=(-1)^K(U_n^{4K-r})^T$, for $1\le r\le 4K-1$. More generally, $U_n^r=(-1)^{Kw}(U_n^{4Kw-r})^T$, for $1\le r\le 4Kw-1$.
\end{enumerate}
\end{prop}
\begin{proof}
Part (b) follows from (a) and unitarity (orthogonality) of $U_n$. For part (a), view the doubling map on $[0,1]$ as the left bit shift on a sequence $\{0,1\}^\N$ corresponding to the binary expansion of $x\in[0,1]$. 
If we partition $[0,1]$ into $2^K$ atoms $E_i=(\frac{i}{2^K},\frac{i+1}{2^K})$, $i=0,\ldots,2^K-1$, then we can identify atom $E_i$ with the length $K$ bit string corresponding to the binary expansion of $\frac{i}{2^K}$. Then  $z\in E_i$ iff the first $K$ digits of its binary expansion match the length $K$ bit string for $E_i$. The Markov matrix $P_n$ then takes an atom indexed by $i=(i_1,\ldots,i_K)$ and sends it to the atoms indexed by $(i_2,\ldots,i_{K-1},0)$ and $(i_2,\ldots,i_{K-1},1)$, the result of the left bit shift.
Thus for $1\le\ell\le K$, there is at most one length $\ell$ path from $i=(i_1,\ldots,i_K)$ to $j=(j_1,\ldots,j_K)$, which is described by the sequence $(i_1,\ldots,i_{\ell},j_1,\ldots,j_K)$ and requires $i_{\ell+1},i_{\ell+2},\cdots, i_{\ell+(K-\ell)}=j_1,j_2,\ldots, j_{K-\ell}$.
Note this recovers Lemma~\ref{lem:sp}(d).

Now considering the signs in $U_{2^K}$ and viewing the indices $i,j$ as length $K$ bit strings, if there is an edge $i\to j$, then 
\begin{align*}
(U_{2^K})_{ij}&=2^{-1/2}\begin{cases}
-1,& i_1=0,j_K=1\\
1,&\text{else}
\end{cases}\\
&=2^{-1/2}(-1)^{(1-i_1)j_K}.
\end{align*}
Thus if there is a length $K$ path $\tau:i\to j$, then
\begin{equation}
(U_{2^K}^K)_{ij} = (U_{2^K})_{i\tau_1}(U_{2^K})_{\tau_1\tau_2}\cdots (U_{2^K})_{\tau_{K-1}j} = 2^{-K/2}\prod_{m=1}^K (-1)^{(1-i_m)j_m},
\end{equation} 
since $(\tau_{m-1})_1=i_{m}$ and $(\tau_{m})_K=j_{m}$. 
This is the structure of a tensor product, 
\begin{equation}
U^K_{2^K}=2^{-K/2}\bigotimes_{m=1}^K\begin{pmatrix}1&-1\\1&1\end{pmatrix},
\end{equation}
so that
\begin{equation}\label{eqn:u2-powers}
U_{2^K}^{2K}=\bigotimes_{m=1}^K\begin{pmatrix}0&-1\\1&0\end{pmatrix},\qquad U_{2^K}^{4K}=\bigotimes_{m=1}^K\begin{pmatrix}-1&0\\0&-1\end{pmatrix}=(-1)^KI_{2^K}.
\end{equation}
\end{proof}

\begin{rmk}
Proposition~\ref{prop:u2k} can also be proved (although with significantly more effort) by analyzing paths in the corresponding de Bruijn graph. As in Section~\ref{sec:doubling-even}, possible paths can be described using trees, but the edges in the trees carry a sign to keep track of the negative signs in the matrix $U_{2^K}$.
\end{rmk}

Since the eigenvalues of $U_{2^K}$ are $4K$-th roots of 1 or $-1$, instead of eigenvectors from eigenvalues in an interval $I(n)\subseteq\R/(2\pi\Z)$ like in Theorem~\ref{thm:uweyl}, we are just interested in all the eigenvectors from a single eigenspace. 
A stronger version of Theorem~\ref{thm:uweyl} for this specific case controls the spectral projection onto a single eigenspace.

\begin{thm}[Eigenspace projection when $n=2^K$]\label{prop:d2weyl}
For $n=2^K$, let $U_n$ be as in \eqref{eqn:doubling-pu}, and let $P^{(n,j)}$ be the projection onto its $j$th eigenspace. Let $r(K):\N\to\N$ be any function satisfying $r(K)<K$, $r(K)\to\infty$, and $K-r(K)-\log_2K\to\infty$ as $K\to\infty$. Then there are sets $G_K\subseteq[1:2^K]$ and $GP_K\subseteq[1:2^K]^2$ with 
\begin{align}
\#G_K&\ge2^K\left(1-\frac{4}{2^{K-r(K)}}\right)=2^K\left(1-o\Big(\frac{1}{4K}\Big)\right)\\
\#GP_K&\ge(2^K)^2\left(1-\frac{8}{2^{K-r(K)}}\right)=(2^K)^2\left(1-o\Big(\frac{1}{4K}\Big)\right),
\end{align}
such that the following hold as $K\to\infty$.
\begin{enumerate}[(a)]
\item For $x\in G_K$ and any $j$,
\begin{align}
\left|\|P^{(n,j)}e_x\|_2^2-\frac{1}{4K}\right|\le\frac{1}{4K}\cdot 10\cdot 2^{-r(K)/2}.\label{eqn:d2diag}
\end{align}
\item For pairs $(x,y)\in GP_K$ and any $j$,
\begin{align}
|\langle e_y, P^{(n,j)}e_x\rangle| &\le \frac{10\cdot 2^{-r(K)/2}}{4K}.\label{eqn:d2off}
\end{align}
\end{enumerate}
\end{thm}

Using \eqref{eqn:d2diag} and summing over all $x$, we also obtain:
\begin{cor}[Eigenspace degeneracy]\label{cor:d2weyl}
The degeneracy of each eigenspace of $U_{2^K}$ is $\frac{2^K}{4K}(1+o(1))$.
\end{cor}

Returning to eigenvectors, Theorem~\ref{prop:d2weyl}(a) applied to Corollary~\ref{cor:randomGaussian} shows that taking a random basis within each eigenspace produces approximately Gaussian eigenvectors.

\begin{thm}[Gaussian eigenvectors when $n=2^K$]\label{thm:doub2}
For $K\in\N$, let $(\psi^{(2^K,j)})_{j=1}^{2^K}$ be a random ONB of eigenvectors chosen according to Haar measure in each eigenspace of $U_n$. Let $\mu^{(2^K,j)}:=\frac{1}{2^K}\sum_{x=1}^{2^K}\delta_{\sqrt{2^K}\psi_x^{(2^K,j)}}$ be the empirical distribution of the scaled coordinates of $\psi^{(2^K,j)}$. Then there is $\varepsilon(K)\to0$ so that for any bounded Lipschitz $f:\C\to\C$ with $L\le1$, as $K\to\infty$,
\begin{equation}
\mathbb{P}\left[\max_{j\in\intbrr{1:2^K}}\left|\int f(x)\,d\mu^{(2^K,j)}(x)-\mathbb{E}f(Z)\right|>\varepsilon(K)\right]\to0,
\end{equation}
where $Z\sim N_\C(0,1)$. In particular, each $\mu^{(2^K,j)}$ converges weakly in probability to $N_\C(0,1)$ as $K\to\infty$.
\end{thm}
\begin{proof}[Proof (of Theorem~\ref{thm:doub2})]
Theorem~\ref{prop:d2weyl} shows we can take $d_\ell=\frac{n}{4K}$, and  $A$ in Corollary~\ref{cor:randomGaussian} to be $o(\frac{n}{4K})$.
Then similar to Subsection~\ref{subsec:proof-gvectors}, take 
\begin{equation}
\varepsilon(K)=\max\left(\frac{4A+6}{(d_\ell-A)-1},\frac{K^{1/4}}{n^{1/4}}\right)\to0,
\end{equation}
and note that $Cn\exp\left(-\frac{cn^{1/2}}{K^{1/2}}(1-o(1))\right)\to0$.
\end{proof}

The rest of this section is the proof of Theorem~\ref{prop:d2weyl}.
\subsection{Polynomial for eigenspace projection}
Instead of using trigonometric polynomials to approximate the spectral projection matrix like in the proof of Theorem~\ref{thm:uweyl}, we use a polynomial with zeros at $4K$-th roots of $1$ or $-1$ to get exact formulas.
Let $U_n$ be as in \eqref{eqn:doubling-pu} with $n=2^K$. First consider $K$ even, so $U_n^{4K}=I$ and the eigenvalues of $U_n$ are $4K$-th roots of unity. Since $\frac{x^{4K}-1}{x-1}=1+x+x^2+\cdots+x^{4K-1}$ is zero at all $4K$-th roots of unity except for $x=1$, the polynomial
\begin{equation}
p_{K,j}(x) = 1+\sum_{\ell=1}^{4K-1}(e^{-2\pi i j/(4K)})^\ell x^\ell
\end{equation}
is zero at all $4K$-th roots of unity except for $e^{2\pi ij/(4K)}$, where it takes the value $4K$. 
Writing
\[
U_n=\sum_{\alpha=0}^{4K-1}e^{2\pi i\alpha/(4N)}\sum_{\lambda=e^{2\pi i\alpha/(4K)}}|\psi^{(\lambda)}\rangle\langle\psi^{(\lambda)}|,
\]
the spectral projection onto the eigenspace of $e^{2\pi ij/(4K)}$ is
\begin{align}
P^{(n,j)}= \sum_{\lambda=e^{2\pi ij/(4K)}}|\psi^{(\lambda)}\rangle\langle\psi^{(\lambda)}| &=\nonumber\frac{1}{4K}\cdot p_{K,j}(U_n) 
= \frac{1}{4K}\left(I + \sum_{\ell=1}^{4K-1}(e^{-2\pi ij/(4K)})^\ell U_n^\ell\right).
\end{align}
If $K$ is odd, then $U_n^{4K}=-I$ and the eigenvalues of $U_n$ are $4K$-th roots of $-1$. These are $\exp(i\frac{\pi}{4K}+\frac{2\pi ij}{4K})$ for $j\in\intbrr{0:4K-1}$. 
For notational convenience, let $\gamma(K):=\begin{cases}e^{i\pi/(4K)},&K\text{ odd}\\1,&K\text{ even}\end{cases}$, so we can write for any $K\in\N$,
\begin{equation}\label{eqn:doublingproj}
P^{(n,j)}:=\sum_{\lambda=e^{2\pi ij/(4K)}\gamma(K)}|\psi^{(\lambda)}\rangle\langle\psi^{(\lambda)}| =\frac{1}{4K}\left(I+\sum_{\ell=1}^{4K-1}(e^{-2\pi ij/(4K)}\bar{\gamma(K)})^\ell U_n^\ell\right).
\end{equation}

\subsection{Powers of $U_n$}\label{subsec:powers}

To estimate the matrix elements of \eqref{eqn:doublingproj}, we need some properties on the powers of $U_n$.
Since by Proposition~\ref{prop:u2k}(b), $U_n^{2K+r} = (-1)^K(U_n^{2K-r})^T$ for $r=0,\ldots,2K-1$, to understand all the powers $U_n,U_n^2,\ldots,U_n^{4K-1}$, it is enough to know the powers $U_n^m$ for $m\in\intbrr{1:K}\cup\intbrr{2K:3K}$. 
We will only need to know where the entries of $U_n^m$ are nonzero, which follows from matrix multiplication:
\begin{lem}[Powers up to $K$]\label{lem:powers1}
Let $n=2^K$. For $m\le K$, let $\mathcal{A}_m$ be the set of real $n\times n$ matrices $A$ such that
\[
|A_{ij}|=\begin{cases}1,&j\in\{2^mi,2^mi-1,\ldots,2^mi-(2^m-1)\}\mod 2^K\\
0,&\text{else}\end{cases}.
\]
$\mathcal{A}_m$ consists of matrices whose nonzero entries are $\pm1$ arranged in $2^m$ descending ``staircases'' with steps of length $2^m$. Then for $m\le K-1$ and $A\in \mathcal{A}_m$,
\[
A\cdot\sqrt{2}U_n\in\mathcal{A}_{m+1}.
\]
In particular, since $2^{1/2}U_n\in\mathcal{A}_1$, then for $m\le K$,
\[
2^{m/2}U_n^m \in \mathcal{A}_m.
\]
\end{lem}

Between $2K$ and $3K$, $U_n^m$ has a flipped staircase structure:
\begin{cor}[Powers from $2K$ to $3K$]\label{cor:powers2}
Let $n=2^K$. For $m\le K$, let $\mathcal{B}_m$ be the set of $n\times n$ matrices $B$ such that the matrix $A$ defined by $A_{ij}:=B_{(n-i)j}$ is in $\mathcal{A}_m$. Then for $m\le K-1$ and $B\in\mathcal{B}_m$,
\[
B\cdot\sqrt{2}U_n\in\mathcal{B}_{m+1}.
\]
In particular, using that $U_n^{2K}$ is a ``flipped diagonal'' matrix with nonzero entries $\pm1$, so that $2^{1/2}U_n^{2K+1}\in B_1$, then for $m\in\intbrr{1:K}$,
\[
2^{m/2} U_n^{2K+m}\in\mathcal{B}_m.
\]
\end{cor}
\begin{proof}
If $A_{ij}=B_{(n-i)j}$, then
$
(BU_n)_{(n-i)j} = \sum_{\ell=1}^n A_{i\ell}(U_n)_{\ell j} = (AU_n)_{ij},
$
and since $\sqrt{2}\cdot AU_n\in\mathcal{A}_{m+1}$, then $\sqrt{2}\cdot BU_n\in\mathcal{B}_{m+1}$.
That $U_n^{2K}$ is a ``flipped diagonal'' matrix with nonzero entries $\pm1$ along the flipped diagonal $(i,n-i)$ follows from  equation~\eqref{eqn:u2-powers}.
Then the matrix $A$ defined by
\[
A_{ij}:=(U_n^{2K}\cdot 2^{1/2}U_n)_{(n-i)j} = 2^{1/2}\sum_{\ell=1}^n\pm\delta_{i,\ell}(U_n)_{\ell j}=(\pm1)2^{1/2}(U_n)_{ij}
\]
is in $\mathcal{A}_1$ so $2^{1/2}U_n^{2K+1}\in\mathcal{B}_1$.
\end{proof}

\subsection{Removing potentially bad points}

This mirrors Subsection~\ref{subsec:badpoints} from the proof of Theorem~\ref{thm:uweyl}, although due to the structure of $U_n^\ell$ here, we consider $1\le\ell\le 4K$ instead of just $1\le \ell\le {\K}+1$. (Figure~\ref{fig:badcoords}.)

Let the set of potentially bad coordinates be 
\begin{equation}
B_K:=\{x\in\intbrr{n}:(U_n^m)_{xx}\ne0\text{ for some }m\in\intbrr{1:r(K)}\cup\intbrr{2K-r(K):2K}\}.
\end{equation} 

\begin{figure}[!ht]
\begin{tikzpicture}[scale=1.7]
\draw[dashed] (0,0)--(1,0);
\draw (1,0)--(3,0);
\draw[dashed] (3,0)--(5,0);
\draw (5,0)--(7,0);
\draw[dashed] (7,0)--(8,0);
\draw (0,.15)--(0,-.15) node[below] {$1$};
\draw (8,.15)--(8,-.15) node[below] {$4K$};
\draw (4,.15)--(4,-.15) node[below] {$2K$};
\draw (2,.1)--(2,-.1) node[below] {$K$};
\draw (6,.1)--(6,-.1) node[below] {$3K$};
\foreach \i in {1,3,5,7}
{
\draw (\i,.1)--(\i,-.1);
}
\foreach \i in {0,3,4,7}
{
\draw [decorate, decoration={brace,raise=7pt,amplitude=4pt}] (\i,0)--(\i+1,0);
\node [above] at (\i+.5,.2) {$r(K)$};
}
\node[left] at (-.08,-.5) {nonzero};
\node[left] at (-.1,-.75) {$|(U_n^m)_{ij}|$ :};
\node[right] at (-.2,-.7) {$2^{-\frac{1}{2}},\;\;\;\;\ldots,\;\;\;2^{-\frac{K-1}{2}},2^{-\frac{K}{2}},2^{-\frac{K-1}{2}},\;\;\;\ldots,2^{-\frac{1}{2}},1$};
\node[right] at (4,-.7) {$, 2^{-\frac{1}{2}},\;\;\ldots,\;\;2^{-\frac{K-1}{2}},2^{-\frac{K}{2}},2^{-\frac{K-1}{2}},\;\ldots,2^{-\frac{1}{2}},1$};
\end{tikzpicture}
\caption{Eliminating bad coordinates in regions where the nonzero entries of $U_n^m$ are large. By Proposition~\ref{prop:u2k}(b), we only need to consider powers up to $2K$ in the definition of $B_K$, since the powers reflect across $2K$.}\label{fig:badcoords}
\end{figure}
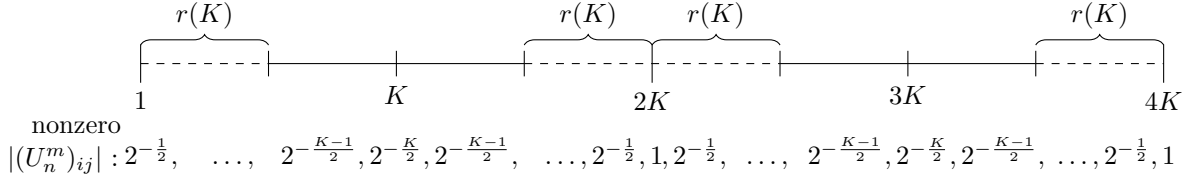
Indexing the $2^K$ atoms of $[0,1]$ by length $K$ bit strings as in the proof of Proposition~\ref{prop:u2k}, we see that for $1\le\ell\le K$, the entry $(U_n^\ell)_{xx}$ is nonzero iff $x$ is of the form $(x_1,\ldots,x_\ell,x_1,\ldots,x_\ell,x_1,\ldots)$, that is $x$ corresponds to a periodic orbit of length $\ell$. There are $2^\ell$ choices for the sequence $(x_1,x_2,\ldots,x_\ell)$, so the diagonal of $U_n^\ell$ contains exactly $2^\ell$ nonzero entries
for $\ell\in\intbrr{1:r(K)}$.
Additionally, by the staircase structure from Corollary~\ref{cor:powers2}, the diagonal of $U_n^\ell$ has at most $2^{2K-\ell}$ nonzero entries for $\ell\in\intbrr{2K-r(K):2K-1}$. Thus 
\begin{equation}
\#B_K\le 2\sum_{\ell=1}^{r(K)}2^\ell=4(2^{r(K)}-1) = o(2^K/K).
\end{equation}
Let the set of good coordinates be $G_K:=\intbrr{n}\setminus B_K$. For $x\in G_K$, then $(U_n^\ell)_{xx}=0$ for $\ell\in\intbrr{1:r(K)}\cup\intbrr{2K-r(K):2K}$ (and also for $\ell\in\intbrr{2K:2K+r(K)}\cup\intbrr{4K-r(K):4K}$), so that for any $j\in\intbrr{0:4K-1}$,
\begin{align*}
\|P^{(n,j)}e_x\|_2^2 &= \frac{1}{4K}\Bigg(1+\sum_{\ell=r(K)+1}^{2K-r(K)-1} (e^{-2\pi ij/(4K)}\bar{\gamma(K)})^\ell (U_n^\ell)_{xx}+\\
&\qquad\qquad\qquad\qquad\qquad\qquad\qquad+\sum_{\ell=2K+r(K)+1}^{4K-r(K)-1} (e^{-2\pi ij/(4K)}\bar{\gamma(K)})^\ell (U_n^\ell)_{xx}\Bigg)\\
&=\frac{1}{4K}(1+\mathcal{O}(2^{-r(K)/2})),\numberthis
\end{align*}
since 
\begin{align}
\left|\sum_{\ell=r(K)+1}^{2K-r(K)-1} (e^{-2\pi ij/(4K)}\bar{\gamma(K)})^\ell (U_n^\ell)_{xx}\right| &\le 2\sum_{\ell=r(K)+1}^K 2^{-\ell/2} \le 10\cdot 2^{-r(K)/2},
\end{align}
and similarly for the second sum. This proves \eqref{eqn:d2diag}.

\subsection{Removing potentially bad pairs of points}

Let the set of potentially bad {pairs} of coordinates be
\begin{multline}
BP_K:=\{(x,y)\in\intbrr{n}^2, x\ne y: (U_n^\ell)_{xy}\ne 0
\text{ for some }\\
\ell\in\intbrr{1:r(K)}\cup\intbrr{2K-r(K):2K+r(K)}\cup\intbrr{4K-r(K):4K-1}\}.
\end{multline}
The matrix $U_n^\ell$ contains $2^\ell\cdot n$ nonzero entries ($n$ entries in each staircase and $2^\ell$ staircases) for $\ell\in\intbrr{1:r(K)}$, and $2^{2K-\ell}\cdot n$ nonzero entries for $\ell\in\intbrr{2K-r(K):2K-1}$ (and the same for flipping $\ell$ across $2K$). Then
\begin{equation}
\#BP_K \le 4\sum_{\ell=1}^{r(K)}2^\ell\cdot n= 8(2^{r(K)}-1)\cdot n = o(n^2),
\end{equation}
and for good pairs $(x,y)\in GP_K:=(\intbrr{n}^2\setminus\{(x,y):x=y\})\setminus BP_K$,
\begin{align}
|\langle e_y, P^{(n,j)}e_x\rangle| &= \left|\frac{1}{4K}\left(\left(\sum_{m=r(K)+1}^{2K-r(K)-1}+\sum_{m=2K+r(K)+1}^{4K-r(K)-1}\right)(e^{-2\pi ij/(4K)}\bar{\gamma(K)})^m(U_n^\ell)_{xy}\right)\right| \nonumber\\
&\le 10\cdot\frac{2^{-r(K)/2}}{4K},
\end{align}
by the same estimates as before. \qed

%%%%%%%%%%%%%%%%%%%%%%%%%%%%%%%%%%%%%%%%%%%%%%%%%%%%%%%%%%
\section{Coordinates that fail the pointwise Weyl law}\label{sec:coordfail}
%%%%%%%%%%%%%%%%%%%%%%%%%%%%%%%%%%%%%%%%%%%%%%%%%%%%%%%%%%

We give a specific example with the doubling map where not all coordinates satisfy the pointwise Weyl law \eqref{eqn:I}.
For $n\in2\N$, let 
\[
U_n=\frac{1}{\sqrt{2}}
\left(
\begin{smallmatrix}1&-1 & & & &&\\
& &1 & -1 &&&\\
&&&&\ddots &&\\
&&&&&1&-1\\
1&1&&&&&\\
&&1&1&&&\\
&&&&\ddots&&\\
&&&&&1&1
\end{smallmatrix}\right).
\]
Letting $P^{I(n)}$ be the spectral projection matrix of $U_n$ onto the arc $I(n)=[-\pi/2,\pi/2]$ on the unit circle, we will show that
\begin{equation}
(P^{I(n)})_{11} \ge 0.89182655+o(1) \ne \frac{1}{2}(1+o(1)).
\end{equation}
Thus for $I(n)=[-\pi/2,\pi/2]$, the sequence of coordinates $(x_n)_n$ with just $x_n=1$ does not satisfy the pointwise Weyl law. Note the coordinate $x=1$ was one of the ``bad'' points that was removed during the proof of the pointwise Weyl law, since it always has the very short periodic loop consisting of just itself.

To approximate $(P^{I(n)})_{11}$ from below, we use the piecewise linear approximation $h_\Delta$ on $\R/(2\pi \Z)$ in Figure~\ref{fig:lin-approx}, defined by
\[
h_\Delta(x)=\begin{cases}
1,&-\frac{\pi}{2}+\Delta\le x\le\frac{\pi}{2}-\Delta\\
\frac{1}{\Delta}\left(x+\frac{\pi}{2}\right),&-\frac{\pi}{2}\le x\le-\frac{\pi}{2}+\Delta\\
-\frac{1}{\Delta}\left(x-\frac{\pi}{2}\right),&\frac{\pi}{2}-\Delta\le x\le\frac{\pi}{2}\\
0,&|x|\ge\frac{\pi}{2}
\end{cases}.
\] 
\begin{figure}[ht]
\centering
\begin{tikzpicture}[scale=.8]
\draw [->] (-3,0)--(3,0);
\draw [line width=1pt] (-3,0)--(-1.5,0)--(-1,1)--(1,1)--(1.5,0)--(3,0);
\draw [dashed] (1,1)--(1,0);% node[below] {$\frac{M}{2}$};
\draw [dashed] (-1,1)--(-1,0);% node[below] {$-\frac{M}{2}$};
\draw (1.5,.1)--(1.5,-.1) node[below] {$\frac{\pi}{2}$};
\draw (-1.5,.1)--(-1.5,-.1) node[below] {$-\frac{\pi}{2}$};
\draw (0,.1)--(0,-.1) node[below] {$0$};
\draw [decorate, decoration={brace,raise=7pt,amplitude=4pt}] (-1.5,1)--(-1,1);
\draw [decorate, decoration={brace,raise=7pt,amplitude=4pt}] (1,1)--(1.5,1);
\node[above] at (-1.25,1.4) {$\Delta$};
\node[above] at (1.25,1.4) {$\Delta$};
\draw [decorate, decoration={brace,raise=7pt,amplitude=4pt}] (-.9,1)--(.9,1);
\node[above] at (0,1.4) {$\pi-2\Delta$};
\end{tikzpicture}
\caption{Plot of $h_\Delta$.}\label{fig:lin-approx}
\end{figure}

Since this $I(n)=[-\pi/2,\pi/2]$ is not shrinking, we do not need further smoothness of the approximation, and the piecewise linear $h_\Delta$ has Fourier coefficients that are easy to work with. We only need continuity and absolutely summable Fourier coefficients, so that the Fourier series converges uniformly to $h_\Delta$.
For convenience we also use the same notation $h_\Delta$ or $\Chi_{[-\pi/2,\pi/2]}$ to denote the corresponding function on the unit circle in $\C$ (via $\R/2\pi\Z\ni t\leftrightarrow e^{it}\in\mathbb{S}^1$). Since pointwise $h_\Delta(t)\le \Chi_{[-\pi/2,\pi/2]}(t)$ for any $\Delta\ge0$, then by the spectral theorem $(h_\Delta(U_n))_{xx}\le \Chi_{[-\pi/2,\pi/2]}(U_n)_{xx}=(P^{I(n)})_{xx}$ for any coordinate $x\in\intbrr{n}$.

To approximate $h_\Delta(U_n)$, we compute its Fourier coefficients $(\hat{h}_{\Delta})_j=\frac{1}{2\pi}\int_{-\pi}^\pi h_{\Delta}(x)e^{-ijx}\,dx$,
\begin{align}
(\hat{h}_{\Delta})_j &= \frac{2}{\pi j^2 \Delta}\sin\left(\frac{j(\pi-\Delta)}{2}\right)\sin\left(\frac{j\Delta}{2}\right),\quad j\ne 0\\
(\hat{h}_{\Delta})_0&=\frac{\pi-\Delta}{2\pi}.
\end{align}
Since the Fourier coefficients are absolutely summable, the partial sums $\sum_{j\in\Z}(\hat{h}_{\Delta})_je^{i jx}$ converge uniformly to $h_{\Delta}$, with the $K$th partial sum $S_Kh_{\Delta}(x):=\sum_{|j|\le K}(\hat{h}_{\Delta})_je^{ijx}$ having error bound
\begin{align}\label{eqn:fs-uniform}
\|S_Kh_{\Delta}-h_{\Delta}\|_\infty\le \sum_{j={K+1}}^\infty\frac{4}{\pi j^2\Delta} \le \frac{4}{\pi K\Delta}.
\end{align}
As long as $\Delta\gg K^{-1}$, this is $o(1)$, and then
\begin{align*}
\left|\langle y|S_Kh_\Delta(U_n)-h_\Delta(U_n)|x\rangle\right| &= \left|\sum_{j=1}^n \left(S_Kh_\Delta(e^{i\theta^{(n,j)}})-h_\Delta(e^{i\theta^{(n,j)}})\right)\langle y|\psi^{(n,j)}\rangle\langle\psi^{(n,j)}|x\rangle\right|\\
&\le \|S_K(h_\Delta)-h_\Delta\|_\infty \left(\sum_{j=1}^n|\psi_y^{(n,j)}|^2\right)^{1/2}\left(\sum_{j=1}^n|\psi_x^{(n,j)}|^2\right)^{1/2}\\
&=o(1).
\end{align*}
So then
\[
(P^{I(n)})_{11}\ge (h_\Delta(U_n))_{11} = ((S_Kh_\Delta)U_n)_{11}+o(1).
\]

As usual,  take $K=\lfloor \log_2n\rfloor$, and consider
\begin{align}\label{eqn:fseries11}
(S_Kh_\Delta)(U_n)_{11} &= \frac{\pi-\Delta}{2\pi}+\sum_{j=1}^K \frac{2}{\pi j^2\Delta}\sin\left(\frac{j(\pi-\Delta)}{2}\right)\sin\left(\frac{j\Delta}{2}\right)\left( (U_n)^j_{11}+(U_n)^{-j}_{11}\right).
\end{align}
Take $K^{-1}\ll\Delta\ll K^{-1/2}$, for example $\Delta=K^{-3/4}$. We split up the sum over $j$ in \eqref{eqn:fseries11} into two regions, first from $j=1$ to $\sqrt{K}$, and then from $\sqrt{K}+1$ to $K$. In the first region, $j\Delta\le \sqrt{K}\Delta\ll1$, so we can Taylor expand the sine terms and evaluate the sum. In the second region, the exponential decay from $(U_n)^j_{11}=2^{-j/2}$ (for $1\le j\le\lfloor\log_2 n\rfloor$, there is only the path $1\to1\to\cdots\to1$ that starts and ends at $1$ and has length $j$) will make the sum $o(1)$ as $K\to\infty$.

For $j\Delta\ll1$, we have $\sin(j\Delta/2)=\frac{j\Delta}{2}+\mathcal{O}(j^2\Delta^2)$ and
\[
\sin\left(\frac{j(\pi-\Delta)}{2}\right)=\begin{cases}
-\frac{j\Delta}{2}+\mathcal{O}(j^3\Delta^3),&j=0\mod 4\\
1-\frac{j^2\Delta^2}{4}+\mathcal{O}(j^4\Delta^4),&j=1\mod 4\\
\frac{j\Delta}{2}+\mathcal{O}(j^3\Delta^3),&j=2\mod 4\\
-1+\frac{j^2\Delta^2}{4}+\mathcal{O}(j^4\Delta^4),&j=3\mod 4
\end{cases}.
\]
Thus
\begin{align*}
(S_Kh_\Delta)(U)_{11} &= \frac{1}{2}-\mathcal{O}(\Delta)+\sum_{j=1}^{\sqrt{K}} \frac{2}{\pi j^2\Delta}\sin\left(\frac{j(\pi-\Delta)}{2}\right)\sin\left(\frac{j\Delta}{2}\right)2\cdot 2^{-j/2}+\sum_{j=\sqrt{K}+1}^K\mathcal{O}(2^{-j/2})\\
&=\frac{1}{2}+
\sum_{\substack{j=1\\j\text{ odd}}}^{\sqrt{K}}
\frac{2}{\pi j}\left((-1)^{((j\text{ mod } 4)-1)/2}+\mathcal{O}(j\Delta)\right)\cdot 2^{-j/2}
+\sum_{\substack{j=1\\j\text{ even}}}^{\sqrt{K}}\frac{\mathcal{O}(j^2\Delta^2)}{j^2\Delta}2^{-j/2}
+o(1)\\
&= \bigg(\frac{1}{2}+ \frac{2}{\pi}\sum_{\substack{j=1\\j=1\mod 4}}^{\sqrt{K}} \left(\frac{1}{j2^{j/2}}-\frac{1}{(j+2)2^{(j+2)/2}}\right)\bigg) +o(1)\\ % last possibly extra negative term is o(1)
&=\bigg(\frac{1}{2}+\frac{1}{\pi}\sum_{\ell=0}^{(\sqrt{K}-1)/4}\frac{5+4\ell}{(1+4\ell)(3+4\ell)2^{(1+4\ell)/2}}\bigg)+o(1).
\end{align*}
Numerically, 
\[
\frac{1}{\pi}\sum_{\ell=0}^{\infty}\frac{5+4\ell}{(1+4\ell)(3+4\ell)2^{(1+4\ell)/2}}\approx 0.39182655,
\]
so that
\begin{equation}
(P^{I(n)})_{11}\ge (S_Kh_\Delta)(U)_{11}+o(1)\ge 0.89182655-o(1) \ne \frac{1}{2}(1+o(1)).
\end{equation}
% looks good numerically (numpy binplot2, 
%n
%4000- (999,3001), 0.8927445
%6000- (1500,4500), 0.893248
%6500- (1625,4874), 0.8924644
%7000- (1750,5250), 0.8917926
%8000- (2000,6000), 0.89135288
%9990- (2497,7492), 0.8910409
%10000- (2500,7500), 0.8916858
%10052- (2512,7539), 0.89113749
%10054- (2512,7541), 0.89249050
%12000- (2999,9001), 0.892155577

\begin{rmk}
A similar statement can be shown with phases $e^{i\Phi}U_n$ for the interval $[t_0-\frac{\pi}{2},t_0+\frac{\pi}{2}]$ for a $t_0$ depending on $\Phi$. The only difference in \eqref{eqn:fseries11} is that the $\left((U_n)^j_{11}+(U_n)^{-j}_{11}\right)$ term becomes $\left(e^{-it_0j}(e^{i\Phi}U_n)^j_{11}+e^{it_0j}(e^{i\Phi}U_n)^{-j}_{11}\right)$. For $j\le\lfloor\log_2n\rfloor$, $(e^{i\Phi}U_n)^j_{11}=e^{ij\Phi_1}2^{-j/2}$ and $(e^{i\Phi}U_n)^{-j}_{11}=e^{-ij\Phi_1}2^{-j/2}$, so taking $t_0=\Phi_1$ reduces this Fourier series back to just \eqref{eqn:fseries11}.

\end{rmk}

\appendix

\section{Details for Section~\ref{sec:gaussian}}\label{sec:details}

In this section we provide details concerning Theorem~\ref{thm:complex-projection}.
The following theorem from \cite{Meckes}, also using \cite{ChatterjeeMeckes}, is a quantitative version of the theorem from \cite{DiaconisFreedman}. Theorem~\ref{thm:complex-projection} will follow from this theorem applied to our case with complex projections.

\begin{thm}[Complex version of Theorem 2 in \cite{Meckes}]\label{thm:MeckesC}
Let $\{x_j\}_{j=1}^n$ be deterministic vectors in $\C^d$, normalized so that $\sigma^2=\frac{1}{nd}\sum_{i=1}^n\|x_i\|^2=1$ and suppose
\begin{align}
\frac{1}{n}\sum_{i=1}^n\left|\|x_i\|^2-d\right|&\le A\\
\sup_{\theta\in\mathbb{S}_\C^{d-1}}\frac{1}{n}\sum_{i=1}^n|\langle\theta,x_i\rangle|^2&\le B.\label{eqn:B}
\end{align}
For a point $\theta\in \mathbb{S}_\C^{d-1}\subset\C^d$, define the measure
$\mu^{(n)}_\theta:=\frac{1}{n}\sum_{j=1}^n\delta_{\langle\theta,x_j\rangle}$ on $\C$. There are absolute numerical constants $C,c>0$ so that for $\theta\sim\operatorname{Unif}(\mathbb{S}_{\C^{d-1}})$, 
 any bounded $L$-Lipschitz $f:\C\to\C$, and $\varepsilon>\frac{2L(A+3)}{d-1}$, there is the quantitative bound
\begin{equation}
\mathbb{P}\left[\left|\int f(x)\,d\mu^{(n)}_\theta(x)-\mathbb{E}f( Z)\right|>\varepsilon\right]\le C\exp\Big(-\frac{c\varepsilon^2d}{L^2B}\Big),
\end{equation}
where $Z\sim N_\C(0,1)$. In particular, if $d=d(n)\to\infty$ as $n\to\infty$, and $A=o(d)$ and $B=o(d)$, then $\mu^{(n)}$ converges weakly in probability to $N_\C(0,1)$ as $n\to\infty$.
\end{thm}
\begin{proof}[Proof outline of Theorem~\ref{thm:MeckesC}]
The proof is the same as the real version in \cite{Meckes}, except that the (multi-dimensional) Theorem~\ref{thm:mv} written below from \cite{ChatterjeeMeckes} replaces the single-variable version. 
The proof idea from \cite{Meckes} is to let  $F(\theta):=\frac{1}{n}\sum_{i=1}^nf(\langle\theta,x_i\rangle)$ and write
\begin{align*}
\mathbb{P}\left[\left|F(\theta)-\mathbb{E}f( Z)\right|>\varepsilon\right] &\le \mathbb{P}\left[\left|F(\theta)-\mathbb{E}F(\theta)\right|>\varepsilon-|\mathbb{E}F(\theta)-\mathbb{E}f(Z)|\right].
\end{align*}
Then one uses Theorem~\ref{thm:mv}, a generalization of Stein's method of exchangeable pairs for abstract normal approximation, to bound $|\mathbb{E}F(\theta)-\mathbb{E}f(Z)|$ with $W=\langle\theta,x_I\rangle$ where $I\sim\operatorname{Unif}\intbrr{n}$, and then one can apply Gaussian concentration (Lemma~\ref{lem:complexsphere}, see for example \cite[\S5.4]{conc-ineq}) to $F$ which is $(L\sqrt{B})$-Lipschitz.
\end{proof}

In the following, $\mathcal{L}(X)$ will denote the law of a random variable or vector $X$.
\begin{thm}[Theorem 2.5 for $\C$ in \cite{ChatterjeeMeckes}]\label{thm:mv}
Let $W$ be a $\C$-valued random variable and for each $\varepsilon>0$ let $W_\varepsilon$ be a random variable such that $\mathcal{L}(W)=\mathcal{L}(W_\varepsilon)$, with the property that $\lim_{\varepsilon\to0}W_\varepsilon=W$ almost surely. 
Suppose there is a function $\lambda(\varepsilon)$ and complex random variables $\Gamma,\Lambda$ such that as $\varepsilon\to0$,
\begin{enumerate}[(i)]
\item $\frac{1}{\lambda(\varepsilon)}\mathbb{E}[(W_\varepsilon-W)|W]\xrightarrow{L^1}-W$.
\item $\frac{1}{2\lambda(\varepsilon)}\mathbb{E}[(W_\varepsilon-W|^2|W]\xrightarrow{L^1}1+\mathbb{E}[\Gamma|W]$.
\item $\frac{1}{2\lambda(\varepsilon)}\mathbb{E}[(W_\varepsilon-W)^2|W]\xrightarrow{L^1}\mathbb{E}[\Lambda|W]$.
\item $\frac{1}{\lambda(\varepsilon)}\mathbb{E}|W_\varepsilon-W|^3\to0$.
\end{enumerate}
Then letting $Z\sim N_\C(0,1)$,
\begin{equation}
d_\mathrm{Wass}(W,Z)\le\mathbb{E}|\Gamma|+\mathbb{E}|\Lambda|,
\end{equation}
where $d_\mathrm{Wass}$ is the Wasserstein distance $d_\mathrm{Wass}(W,Z)=\sup_{\|g\|_\mathrm{Lip}\le 1}|\mathbb{E}g(W)-\mathbb{E}g(Z)|$.
\end{thm}

\begin{lem}[Gaussian concentration on the complex sphere]\label{lem:complexsphere}
Let $F:\C^d\to\C$ be $L$-Lipschitz and $\theta\sim\operatorname{Unif}(\mathbb{S}_\C^{d-1})$. Then there are absolute constants $C,c>0$ so that
\begin{equation}
\mathbb{P}[|F(\theta)-\mathbb{E}F(\theta)|\ge t]\le C\exp(-ct^2d/L^2).
\end{equation}
\end{lem}

\subsection{Proof of Theorem~\ref{thm:complex-projection}}
Let $v_1,\ldots,v_d$ be an orthonormal basis for $V^{(\nu)}$, and let $M_V$ be the $n\times d$ matrix with those vectors as columns. 
Apply Theorem~\ref{thm:MeckesC} to the data set $\sqrt{n}M_V^*e_1,\sqrt{n}M_V^*e_2, \ldots,\sqrt{n}M_V^*e_n$ in $\C^d$, noting that since $P^{(\nu)}=M_VM_V^*$, then
\begin{align*}
\sigma^2=\frac{1}{nd}\sum_{x=1}^nn\|M_V^*e_x\|_2^2=\frac{1}{d}\sum_{x=1}^N\langle e_x|P^{(\nu)}|e_x\rangle=\frac{1}{d}\operatorname{tr}P^{(\nu)}=1.
\end{align*}
We can also take $B=1$ since for any $\theta\in\mathbb{S}_\C^{d-1}$,
\[
\frac{1}{n}\sum_{x=1}^n|\langle\theta,\sqrt{n}M_V^*e_x\rangle|^2=\sum_{x=1}^n\langle M_V\theta,e_x\rangle\langle e_x,M_V\theta\rangle=\|M_V\theta\|_{\C^n}^2=%\langle \theta,M_V^*M_V\theta\rangle_{\C^d}=
\|\theta\|_{\C^d}^2.
\]
If $\theta$ is uniform on $\mathbb{S}_\C^{d-1}\subset\C^d$, then $M_V\theta$ is uniform on $\mathbb{S}(V^{(\nu)})$, so $\frac{1}{n}\sum_{j=1}^n\delta_{\langle\theta,\sqrt{n}M_V^*e_x\rangle}\sim\frac{1}{n}\sum_{j=1}^n\delta_{\sqrt{n}\omega_x}$.

\bibliographystyle{siamnodashnocaps}
\bibliography{./intmap_ref.bib}

\end{document}